\documentclass[amssymb,nobibnotes,superscriptaddress,longbibliography,notitlepage,onecolumn]{revtex4-1}

\usepackage{enumerate,amsmath,amsthm}
\usepackage{graphicx}
\usepackage{braket}
\usepackage{calc}
\usepackage{accents}
\usepackage{tabularx,lipsum,environ,amsmath,amssymb}

\usepackage{comment}

\usepackage[hang,small,bf]{caption}
\usepackage[subrefformat=parens]{subcaption}
\usepackage[hidelinks]{hyperref}

\usepackage[running]{lineno}

\usepackage{algorithm}
\usepackage{algpseudocode}

\algnewcommand\algorithmicforeach{\textbf{for each}}
\algdef{S}[FOR]{ForEach}[1]{\algorithmicforeach\ #1\ \algorithmicdo}
\usepackage{multirow}

\newtheorem{definition}{Definition}

\newtheorem{theorem}{Theorem}
\newtheorem{corollary}{Corollary}
\newtheorem{lemma}{Lemma}

\makeatletter
\newcommand{\problemtitle}[1]{\gdef\@problemtitle{#1}}
\newcommand{\probleminput}[1]{\gdef\@probleminput{#1}}
\newcommand{\problemquestion}[1]{\gdef\@problemquestion{#1}}
\NewEnviron{problem}{
  \problemtitle{}\probleminput{}\problemquestion{}
  \BODY
  \par\addvspace{.5\baselineskip}
  \noindent
  \begin{tabularx}{\textwidth}{@{\hspace{\parindent}} l X c}
    \multicolumn{2}{@{\hspace{\parindent}}l}{\@problemtitle} \\
    \textbf{Input:} & \@probleminput \\
    \textbf{Output:} & \@problemquestion
  \end{tabularx}
  \par\addvspace{.5\baselineskip}
}
\makeatother

\begin{document}


\title{Quantum State Preparation via Free Binary Decision Diagram}%

\author{Yu Tanaka}
\email{yukanata@gmail.com}
\affiliation{
Advanced Research Laboratory, Digital {\textnormal\&} Technology Platform I{\textnormal\&}S,
Sony Group Corporation, 1-7-1 Konan, Minato-ku, Tokyo, 108-0075, Japan
}%
\affiliation{%
Department of Physics, Graduate School of Science, The University of Tokyo,
Hongo 7-3-1, Bunkyo-ku, Tokyo 113-0033, Japan
}%

\author{Hayata Yamasaki}

\affiliation{%
Department of Physics, Graduate School of Science, The University of Tokyo,
Hongo 7-3-1, Bunkyo-ku, Tokyo 113-0033, Japan
}%

\author{Mio Murao}

\affiliation{%
Department of Physics, Graduate School of Science, The University of Tokyo,
Hongo 7-3-1, Bunkyo-ku, Tokyo 113-0033, Japan
}%

\date{\today}

\begin{abstract}
    Quantum state preparation (QSP) is the task of preparing a quantum state from a given classical description.
    The classical description of an $n$-qubit quantum state may have $\exp(O(n))$ parameters in general, which makes preparing the corresponding state inherently inefficient in the worst case.
    However, in many practical cases, we may be able to employ suitable data structures for QSP. 
    Ordered binary decision diagrams (OBDDs) and free BDDs (FBDDs) are data structures that can represent large-scale data in a compressed way.
    An efficient QSP method for a subclass of OBDDs is known, but requires an $O(2^n)$-sized quantum circuit in general, while QSP based on FBDDs, which includes OBDDs as a special case, remains unexplored.
    Here, we construct a state-preparation procedure for QSP when the classical description of a quantum state is given by an FBDD with weighted edges and analyze the space and time complexity of QSP in this setting.
    We provide a nontrivial example of an $n$-qubit state that can be represented by a weighted FBDD with $N=O(\mathrm{poly}(n))$ nodes rather than $\mathrm{exp}(O(n))$. 
    We show that any quantum state represented by a QSP-admissible weighted FBDD with $N$ nodes can be prepared by an $O(N)$-sized quantum circuit using at most $N$ ancillary qubits, exponentially improving the required circuit size for QSP compared to other BDD-based QSPs\@.
    We also provide another example of an $n$-qubit state that can be represented by a weighted FBDD with $N=O(n^2)$ nodes and $O(n^2)$ ancillary qubits, for which the known amplitude-amplification-based QSP cost expression is exponential in $n$.
    These results provide techniques to employ FBDDs as a tool for broadening the possibility of efficient QSP\@.
\end{abstract}

\maketitle

\section{Introduction}

Quantum state preparation (QSP) is the task of preparing a quantum state from a given classical description of the state, which may be given by oracle access to the coefficients of the quantum state or by some data structure to store the classical description of the quantum state.
QSP is a prerequisite input step for major quantum algorithms such as Harrow-Hassidim-Lloyd (HHL) algorithm~\cite{Harrow2009-ci}, Hamiltonian simulation~\cite{Low2019hamiltonian,10.1145/3313276.3316366}, and quantum machine learning with optimized random features~\cite{Yamasaki2020-xh,yamasaki2022exponential,pmlr-v202-yamasaki23a}.
In many quantum algorithms, data-loading or state-preparation subroutines can dominate the overall resource cost and may therefore limit the achievable quantum speedup~\cite{Aaronson2015-tw}. 
Therefore, the computational complexity of QSP is of great interest from both practical and fundamental viewpoints.

In general, the classical description of an $n$-qubit quantum state may have $\exp(O(n))$ parameters, which are inherently hard to deal with in polynomial computational resources in $n$.
In Refs.~\cite{Bergholm2005-to, Plesch2011-mx}, it is shown that an $n$-qubit quantum state can be prepared by an $O(2^n)$-depth quantum circuit composed of single-qubit unitaries and \textsc{CNOT} gates, without ancillary qubits. 
References~\cite{Zhang2022-bb,Rosenthal2021-pj,10.1109/TCAD.2023.3244885} gave an $O(n)$-depth quantum circuit to prepare an $n$-qubit quantum state using single- and two-qubit gates, but the cost of the required number of ancillary qubits grows exponentially in $n$. 
The trade-off between the circuit depth and the number of ancillary qubits for QSP is studied in Ref.~\cite{Rosenthal2021-pj}.

On the other hand, in many cases, we can employ suitable data structures to represent the coefficients of a quantum state in a compressed way, e.g., by using polynomially many nonzero coefficients~\cite{Zhang2022-bb}, discrete approximations of an efficiently integrable probability density function (Grover and Rudolph~\cite{Grover2002-yz}), a matrix product state (MPS)~\cite{Holmes2020-fe, Garcia-Ripoll2021-ft}, prior knowledge of the coefficients~\cite{Sanders2019-fb}, polynomial or Fourier approximations~\cite{McArdle2022-mu}, and ordered binary decision diagrams (OBDDs)~\cite{Mozafari2020-lq, Mozafari2022-mj}.
In their settings, those results resolve the trade-off between the circuit depth and the number of ancillary qubits for QSP.

Function-loading methods further illustrate how suitable classical structure can make QSP efficient. Reference~\cite{Moosa2024FourierSeriesLoader} introduced the Fourier Series Loader, which exactly prepares states specified by truncated Fourier-series representations using circuits whose depth is linear in the number of qubits in the loaded register and in the number of retained Fourier coefficients. Reference~\cite{GonzalezConde2024PolynomialEncoding} studied amplitude encoding of real polynomial functions through matrix-product-state constructions and a DHWT/QSVT-based construction. These results show that the choice of input representation is crucial: Fourier-series, polynomial, or tensor-network structure can make state preparation efficient for important classes of states. (See Table~\ref{tbl:comparison} for comparison.)

However, in general, QSP with a $\mathrm{poly}(n)$-sized classical description of any $n$-qubit quantum state is known to be computationally hard~\cite{Aharonov2007-kx}.
Therefore, one challenge is to develop a state-preparation procedure for QSP with an $O(\mathrm{poly}(n))$-sized quantum circuit to prepare a wider class of quantum states with $d=2^{\Omega(n)}$ nonzero coefficients.

To address the challenge, we employ a data structure for representing Boolean functions, called binary decision diagrams (BDDs)~\cite{Breitbar_P1995-dg}.
A $\mathrm{BDD}_f$ representing a Boolean function $f:\{0,1\}^n \to \{0,1\}$ is a rooted directed acyclic graph with one or two nodes of outdegree zero and other nodes of outdegree two, where each node in the BDD represents a Boolean variable.
In this paper, our method specifically utilizes a {\it free} binary decision diagram (FBDD), rather than a binary tree (BT) or an OBDD used in standard BDD-based QSP methods.
QSP using FBDDs has not been explored because FBDDs do not share a key property that both BTs and OBDDs have. This property is that the order of occurrences of the Boolean variables on each path from the root to the terminal nodes is consistent. It ensures that, at each depth, all the Boolean variables corresponding to nodes at the same depth represent the same qubit.
Leveraging this property, for example, the method in Ref.~\cite{Prakash2014QuantumAF} achieves quantum state preparation (QSP) by coherently accessing the qubits labeled by nodes at the same depth.

In this paper, we construct a state-preparation procedure for target states specified by compact weighted FBDDs (WFBDDs). A WFBDD combines two kinds of information: its underlying FBDD describes the Boolean support structure, while its complex edge weights determine the amplitudes through the locally normalized quantum-state semantics introduced in Sec.~\ref{Sec3}. Our main result is twofold. First, we show that every WFBDD whose locally normalized representative satisfies our QSP-admissibility condition defines a normalized $n$-qubit quantum state. Second, given such a WFBDD with $N$ nodes, we construct a quantum circuit that prepares the corresponding state using $O(N)$ single- and two-qubit gates and at most $N$ ancillary qubits.

The efficiency claims in this work are conditional on the input model
described above. Specifically, we assume that the relevant Boolean function
or target state is supplied by a compact FBDD or WFBDD description,
respectively. For WFBDD inputs, we assume that the locally normalized
representative satisfies the QSP-admissibility condition of
Def.~\ref{def:qsp_admissible_wfbdd}. In our circuit constructions, the
FBDD/WFBDD description is treated as classical compile-time input from
which the quantum gates are generated. The cost of finding, constructing,
or compressing such a description is not included in the stated circuit
complexity. If one instead implements the diagram operations through
qROM/QRAM-style coherent lookup of node, edge, or weight data during the
quantum computation, the cost of constructing that data-access mechanism
must be accounted for separately. Thus, our results should be understood
as efficient state-preparation procedures for instances in which compact
FBDD/WFBDD descriptions are already available.

{\nolinenumbers
\begin{table}[t]
    \centering
    \begin{tabular}{c|c|c|c|c|c} \hline
          & Depth & \# Ancillary qubits & \begin{tabular}{c} \# Amplitude \\ amplifications \end{tabular} & \begin{tabular}{c} Assumption \\ / Oracle \end{tabular} & \begin{tabular}{c} An $n$-qubit quantum state \\ that can be prepared \end{tabular} \\ \hline
        \begin{tabular}{c} FBDD-based \\ (This work) \end{tabular} & $O({\rm poly}(n))$ & $O({\rm poly}(n))$ & - & \begin{tabular}{c} single- and \\ two-qubit gates \end{tabular} & \begin{tabular}{c} Described by a compact \\ QSP-admissible \\ WFBDD with ${\rm poly}(n)$ nodes. \end{tabular} \\ \hline
        OBDD-based~\cite{Mozafari2020-lq} & $O(2^n)$ & $0$ & \multirow{2}{*}{-} & \multirow{2}{*}{\begin{tabular}{c} Single- and \\ two-qubit gates \end{tabular}} & \begin{tabular}{c} Described by an OBDD \\ with ${\rm poly}(n)$ nodes. \end{tabular} \\
        \cite{Mozafari2022-mj} & $O(kn)$ & $1$ & & & \begin{tabular}{c} / with $k$ paths. \end{tabular} \\ \hline
        Width-first~\cite{Bergholm2005-to, Plesch2011-mx} & $O(2^n)$ & $0$ & - & \begin{tabular}{c} Single-qubit and \\ \textsc{CNOT} gates \end{tabular} & Generally applicable. \\ \hline
        Depth-first~\cite{Zhang2022-bb,Rosenthal2021-pj} & $O(n)$ & $O(2^n)$ & \multirow{2}{*}{-} & \multirow{2}{*}{\begin{tabular}{c} Single- and \\ two-qubit gates \end{tabular}} & \multirow{2}{*}{Generally applicable.} \\
        Balanced~\cite{10.1109/TCAD.2023.3244885} & $O(\frac{2^n}{m+n})$ & $m \in \left[2n, O\left(\frac{2^n}{n\log n}\right)\right]$ & & & \\ \hline
        \begin{tabular}{c} Sparse state \\ preparation~\cite{Zhang2022-bb} \end{tabular} & $O(\log (nd))$ & $O(nd \log d)$ & - & \begin{tabular}{c} Single- and \\ two-qubit gates \end{tabular} & \begin{tabular}{c} The number of nonzero \\ coefficients is $d$. \end{tabular} \\ \hline
        Grover-Rudolph~\cite{Grover2002-yz} & $O(n)$ & $O(g)$ & - & \begin{tabular}{c} Coherent \\ arithmetic \end{tabular} & \begin{tabular}{c} Efficiently integrable \\ probability distribution. \end{tabular} \\ \hline
        \begin{tabular}{c} Matrix product state \\ encoding~\cite{Holmes2020-fe, Garcia-Ripoll2021-ft} \end{tabular} & $O(n)$ & $0$ & - & \begin{tabular}{c} Single- and \\ two-qubit gates \end{tabular} & \begin{tabular}{c} Approximated by \\ $O(1)$ bond dimension. \end{tabular} \\ \hline
        Black-box~\cite{Sanders2019-fb} & $O(g \log_2 g)^{\ast 1}$ & $2g+1$ & \multirow{2}{*}{$O(2^{n/2})^{\ast 2}$} & \multirow{2}{*}{\begin{tabular}{c} Amplitude \\ oracle \end{tabular}} & \multirow{2}{*}{Generally applicable.} \\
        \cite{Bausch2022fastblackboxquantum} & $O(g)^{\ast 1}$ & $g+\log_2 g$ & & & \\ \hline
        QET-based~\cite{McArdle2022-mu} & $O\left(\frac{nd}{\mathcal{F}_f^{[N]}}\right)^{\ast 3}$ & $4$ & - & \begin{tabular}{c} Single- and \\ two-qubit gates \end{tabular} & \begin{tabular}{c} Polynomial or Fourier \\ series approximation. \end{tabular} \\ \hline
        \begin{tabular}{c} Fourier Series \\ Loader~\cite{Moosa2024FourierSeriesLoader} \end{tabular} & $O(n+M_F)^{\ast 4}$ & $0$--$1$ & - & \begin{tabular}{c} Single- and \\ two-qubit gates \end{tabular} & \begin{tabular}{c} Truncated Fourier-series \\ approximations. \end{tabular} \\ \hline
        \begin{tabular}{c} Matrix product state \\ encoding~\cite{GonzalezConde2024PolynomialEncoding} \end{tabular} & \begin{tabular}{c} exact $O(n^2d_p)$;\\ approx. $O(n)$$^{\ast 5}$ \end{tabular} & $0$ & - & \begin{tabular}{c} Single- and\\ two-qubit gates \end{tabular} & \begin{tabular}{c} Real polynomials;\\ compact MPS / small $\chi$ \end{tabular} \\ \hline
        \begin{tabular}{c} DHWT/QSVT~\cite{GonzalezConde2024PolynomialEncoding} \end{tabular} & \begin{tabular}{c} $O(k_0d_p)^{\ast 6}$ \end{tabular} & $O(k_0)$ & $O(1/F)$ & \begin{tabular}{c} Single- and \\ two-qubit gates \end{tabular} & \begin{tabular}{c} Real polynomials;\\ controllable approx. \end{tabular} \\ \hline
        \begin{tabular}{c} Adiabatic state \\ preparation~\cite{Wan2020-ui, Rattew2022-zz} \end{tabular} & $O\left( \frac{g^2}{\mathcal{F}^4_f} \right)^{\ast 7}$ & $O(n+g)$ & - & \begin{tabular}{c} Amplitude \\ oracle \end{tabular} & \begin{tabular}{c} Efficiently computable \\ function. \end{tabular} \\ \hline
    \end{tabular}
    \caption{Comparison with previous QSP methods. In the figure, $n$ denotes the number of qubits, $k$ denotes the number of root-to-terminal paths in an OBDD, and $g$ denotes the bit precision used to represent the coefficients of a quantum state. The number of paths $k$ can be exponential in $n$ even when the OBDD has a compact node representation. For the FBDD-based method, a WFBDD is an FBDD equipped with complex edge weights; for QSP, the associated state is interpreted through a locally normalized representative that is QSP-admissible, as introduced in Secs.~\ref{Sec2} and~\ref{Sec3}. The WFBDD description is assumed to be supplied as input, and qROM/QRAM construction costs are not included in the listed circuit complexity. $\ast 1$ The number of non-Clifford gates used per one amplitude amplification. $\ast 2$ Reference~\cite{Bausch2022fastblackboxquantum} reduces the number of amplitude amplifications to $O(1)$ for a quantum state such that $|| \bar{A} ||_2 / || \alpha ||_1 = O(2^{n/2})$, where $\alpha$ is the amplitude vector, and $\bar{A}$ is the amplitudes' average bit weight vector. $\ast 3$ In the figure, $d$ is the degree of a polynomial, and $\mathcal{F}_f^{[N]}$ is the approximation of the L2-norm filling-fraction $\sqrt{\int_a^b |f(x)|^2 dx} / \sqrt{(b-a)|f|^2_{{\rm max}}}$ for a function $f: [a,b] \to \mathbb{R}$. $\ast 4$ For the Fourier Series Loader, $M_F$ is the number of retained Fourier coefficients, and $n$ denotes the total number of qubits in the loaded register. In Ref.~\cite{Moosa2024FourierSeriesLoader}, a $D$-dimensional function is loaded into $Dn_0$ qubits. The $0$--$1$ ancilla entry corresponds to the periodic case and its non-periodic extension. $\ast 5$ For the MPS-based polynomial-amplitude encoding of Ref.~\cite{GonzalezConde2024PolynomialEncoding}, $d_p$ is the polynomial degree. The approximate $O(n)$ entry assumes a fixed small bond dimension, such as $\chi=2$, and does not include the cost of constructing the MPS representation or decomposing the associated multi-qubit unitaries. $\ast 6$ For the DHWT/QSVT approach, $k_0$ is the number of qubits used in the linear-function loader $U_{L,k_0}$, $d_p$ is the polynomial degree, and $F$ is the filling ratio in amplitude amplification. The construction uses $U_{L,k_0}$, its adjoint, and controlled variants; the total cost including amplitude amplification scales as $O(k_0d_p/F)$. $\ast 7$ Query complexity of the amplitude oracle, where $\mathcal{F}_f$ is the filling ratio $\lVert f \rVert_1 / \lVert f \rVert_{\max}$.} \label{tbl:comparison}
\end{table}
}

When $N = \mathrm{poly}(n)$, our FBDD-based state-preparation construction gives a $\mathrm{poly}(n)$-sized quantum circuit for QSP with $\mathrm{poly}(n)$ ancillary qubits, which reduces the ancillary qubits needed for existing QSP methods for preparing a general class of states~\cite{Zhang2022-bb,10.1109/TCAD.2023.3244885,Rosenthal2021-pj}.
Our FBDD-based state-preparation construction exponentially reduces the number of two-qubit gates used by a QSP method~\cite{Mozafari2020-lq} based on ordered BDDs (OBDDs), a proper subset of FBDDs~\cite{Breitbar_P1995-dg}. 

Furthermore, our FBDD-based state-preparation construction also exponentially reduces the number of two-qubit gates used by another OBDD-based method~\cite{Mozafari2022-mj} preparing an $n$-qubit quantum state described by an OBDD with $k = O(2^n)$ paths.
The approach in Ref.~\cite{Mozafari2022-mj} operates by preparing the paths in
an OBDD one-by-one. 
As a result, the circuit complexity is linear in the number $k$ of paths in the OBDD, rather than in the number $N$ of nodes. It is important to emphasize that, even if $N$ is polynomial in the depth of the BDD, $k$ is not necessarily polynomial. 
In contrast, our method achieves circuit complexity that is linear in the number $N$ of nodes, which presents a significant advantage over the approach.

We give a non-trivial example of exponentially compressible data structures by using a symmetric Boolean function~\cite{Breitbar_P1995-dg}.
This result shows the existence of $\Theta(n)$-qubit quantum states with $d = 2^{\Omega(n)}$ nonzero coefficients that can be prepared by a $\mathrm{poly}(n)$-sized quantum circuit with $\mathrm{poly}(n)$ ancillary qubits using our FBDD-based state-preparation construction, but requires an exponential number of ancillary qubits using the algorithm for QSP in Ref.~\cite{Zhang2022-bb}.

Reference~\cite{Vinkhuijzen2024-to} provides an analytical investigation of three widely used classical descriptions of quantum states: MPS, decision diagrams (DDs), and restricted Boltzmann machines (RBMs). The authors map the relative succinctness of these data structures. 
Our results on weighted FBDDs are closely related to one of their results on the DDs called semiring-labeled decision diagrams (${\rm SLDD}_{\times}$), which is interpreted as weighted OBDDs in our paper's terminology. 
This fact implies that the weighted FBDD can provide a more succinct classical description of quantum states than ${\rm SLDD}_{\times}$.

We discuss four possible applications. 
First, we explore how to prepare a uniform superposition of basis states designated by a Boolean function $f:\{0,1\}^n \to \{0,1\}$.
We define a quantum state encoding the Boolean function $f$ as $\ket{\phi_f} := |f|^{-\frac{1}{2}} \sum_{x \in \{ 0,1 \}^n} f(x) \ket{x}$, {\it i.e.}, the amplitude encoding of the Boolean function, where $|f|$ is the number of $x \in \{ 0,1 \}^n$ such that $f(x) = 1$.
Then, for a given $\mathrm{FBDD}_f$ representing the Boolean function $f$, we propose a classical algorithm that assigns weights to the $\mathrm{FBDD}_f$ using $O(N)$ queries to its nodes to derive a classical description of the quantum state encoding $f$.

Second, we give another encoding of the Boolean function $f$ represented by an $\mathrm{FBDD}_f$. We give a classical algorithm for constructing a quantum circuit mapping $\ket{x}$ to $e^{i\theta f(x)} \ket{x}$ using $O(N)$ Toffoli gates and $N-1$ ancillary qubits.

Third, we provide another example of an $n$-qubit quantum state that can be prepared by an $O(n^2)$-sized quantum circuit using $O(n^2)$ ancillary qubits. We show that the cost expression for the number of amplitude-amplification steps in the black-box QSP algorithm of Ref.~\cite{Bausch2022fastblackboxquantum} is exponential for this state. 

Fourth, we discuss how our proposed techniques for QSP can be applied to the block encoding, a technique of encoding a matrix as a block of a unitary, defined in Ref.~\cite{10.1145/3313276.3316366}.
We give two non-trivial examples of efficiently implementable block encodings of non-sparse matrices.

The rest of the paper is organized as follows. 
In Sec.~\ref{Sec2}, we introduce BDDs, FBDDs, OBDDs, and WFBDDs. We also clarify the Boolean path semantics of FBDDs, the weighted path semantics of WFBDDs, QSP-admissible WFBDDs, and locally normalized representatives.
In Sec.~\ref{Sec3}, we define the quantum-state semantics of a WFBDD via its locally normalized representative and give a state-preparation circuit for the corresponding quantum state.
We also give an example of an $O(\mathrm{poly}(n))$-sized $\mathrm{FBDD}_f$ that provides a compact support description for a quantum state with exponentially many nonzero coefficients.
In Sec.~\ref{Sec4}, we construct, from a given $\mathrm{FBDD}_f$ representing a Boolean function, a QSP-admissible WFBDD for the uniform superposition over the satisfying assignments by using effective branch model counts.
In Sec.~\ref{Sec5}, for a given $\mathrm{FBDD}_f$, we give a classical algorithm for constructing a quantum circuit mapping $\ket{x}$ to $e^{i\theta f(x)}\ket{x}$ directly.
In Sec.~\ref{Sec6}, we provide a $\mathrm{WFBDD}_f$ using $n^2$ nodes. We show that our FBDD-based state-preparation circuit can generate the quantum state associated with the $\mathrm{WFBDD}_f$ under the quantum-state semantics of Sec.~\ref{Sec3}, while the known cost expression for the number of amplitude-amplification steps in the black-box QSP algorithm of Ref.~\cite{Bausch2022fastblackboxquantum} is exponential in $n$.
In Sec.~\ref{Sec7}, we discuss block-encoding examples in cases where the required state-preparation ingredients admit compact WFBDD descriptions.
Finally, we give a conclusion.

{\nolinenumbers
\begin{figure}[t]
    \centering
    \includegraphics[keepaspectratio,scale=0.2]{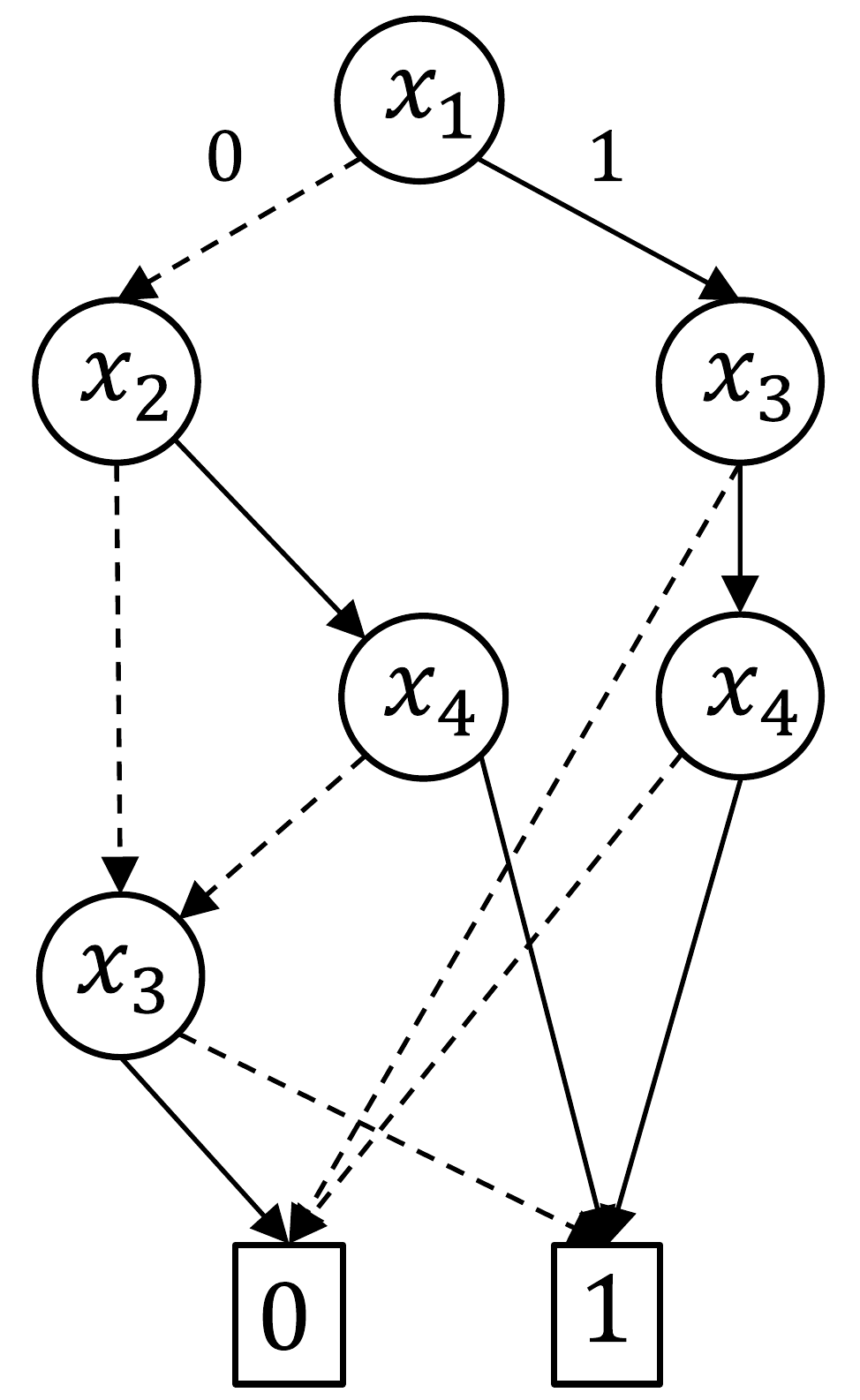}
    \caption{A free binary decision diagram ($\mathrm{FBDD}_f$) representing a Boolean function $f(x_1,x_2,x_3,x_4) = x_1x_3x_4 + \bar{x}_1x_2x_4 + \bar{x}_1\bar{x}_2\bar{x}_3 + \bar{x}_1x_2\bar{x}_3\bar{x}_4$. The $\mathrm{FBDD}_f$ is a rooted directed acyclic graph with one or two rectangular nodes of outdegree zero and the other circular nodes of outdegree two. The rectangular nodes are labeled with $0$ and $1$, respectively, and are called the $0$-terminal node and the $1$-terminal node. The other circular nodes are labeled with the variables $x_1, x_2, x_3$, and $x_4$, called the internal nodes. One node of indegree zero is called the root. Two outgoing edges of each internal node are labeled with $0$ and $1$, respectively, and are called the $0$-edge and the $1$-edge, which are shown as dotted and solid lines. A path from the root to one of the terminal nodes represents a pair of an input bit sequence and the output bit. For given $x_1, x_2, x_3, x_4 \in \{0,1\}$, $f(x_1, x_2, x_3, x_4)$ can be evaluated by starting from the root and selecting the $x_i$-edge in each node until reaching the terminal nodes. If there is no node labeled with a variable on the path, the assignment of the variable is arbitrary. For instance, $x_1,x_2,x_3,x_4=0,0,0,0$ and $0,0,0,1$ share the same path. Note that the word "free" means that the order of occurrences of the variables on each path from the root to the terminal nodes is not fixed.} \label{fig:exam_fbdd}
\end{figure}
}

\section{Preliminaries} \label{Sec2}

In this section, we give some definitions used in this study. 
Let $b$, $[n]$, and $[n]_0$ denote an element of $\{0,1\}$, $\{1,\ldots,n\}$, and $\{0,\ldots,n-1\}$, respectively.
First, we introduce a binary decision diagram (BDD) and its variants, a free binary decision diagram (FBDD) and an ordered binary decision diagram (OBDD). Then, we define an FBDD with weighted edges (WFBDD).

A binary decision diagram (BDD) is a rooted directed acyclic graph that represents a Boolean function $f(x_1,\ldots, x_n)$ of variables $x_1,\ldots, x_n$ and has one or two nodes of outdegree zero, with all other nodes equipped with two labeled outgoing edges. (See Fig.~\ref{fig:exam_fbdd}.)
The nodes of outdegree zero are labeled with $0$ and $1$, respectively, and are called the $0$-terminal node or the $1$-terminal node.
The other nodes are labeled with the variables $x_1,\ldots, x_n$, called the internal nodes.
There is one internal node of indegree zero, called the root.
Two outgoing edges of each internal node are labeled with $0$ and $1$, respectively, and are called the $0$-edge and the $1$-edge of the internal node.
A path from the root to one of the terminal nodes represents a pair of an input bit sequence and the output bit. For given $x_1, \ldots, x_n \in \{0,1\}$, $f(x_1, \ldots, x_n)$ can be evaluated by starting from the root and selecting the $x_i$-edge in each node until reaching the terminal nodes. 
If there is no node labeled with a variable on the path, the assignment of the variable is arbitrary. 

Although a BDD has redundant nodes in general, we can reduce the number of redundant nodes in a given BDD according to two contraction rules, called the redundant node deletion and the equivalent node sharing.
The redundant node deletion is the deletion of an internal node when two outgoing edges from the internal node have the same child node.
The equivalent node sharing is the merging of two internal nodes when the two internal nodes share both of the child nodes of the two outgoing edges from their respective internal nodes.
A BDD to which the contraction rules have been applied exhaustively is called a reduced BDD (RBDD).
(See Fig.~\ref{FIG:BDD_rules}.)
In summary, we give a formal definition of BDDs and RBDDs as follows.

{\nolinenumbers
\begin{figure}[t]
  \begin{minipage}[b]{0.3\linewidth}
    \centering
    \includegraphics[keepaspectratio, scale=0.15]{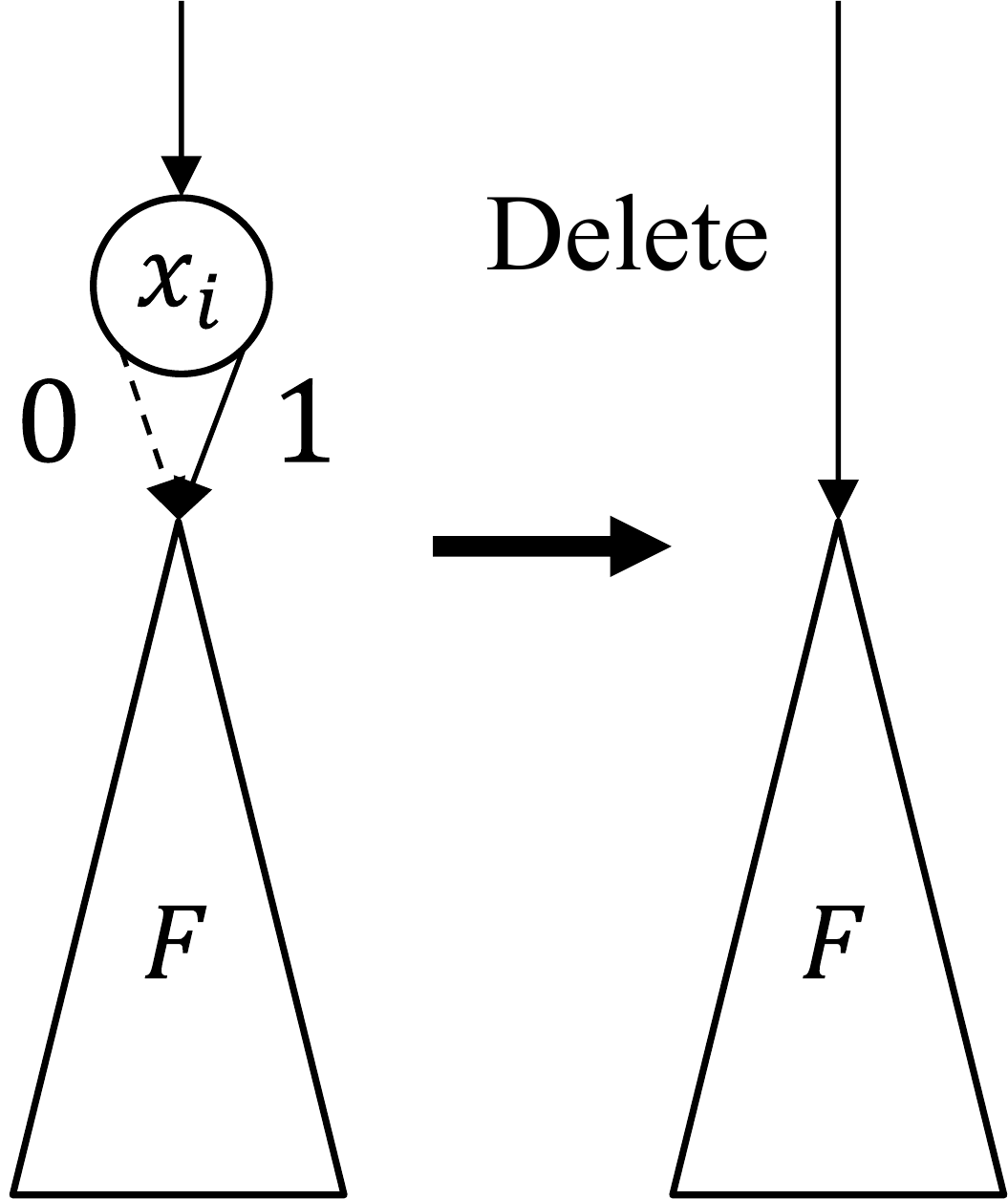}
    \subcaption{}
  \end{minipage}
  \begin{minipage}[b]{0.3\linewidth}
    \centering
    \includegraphics[keepaspectratio, scale=0.15]{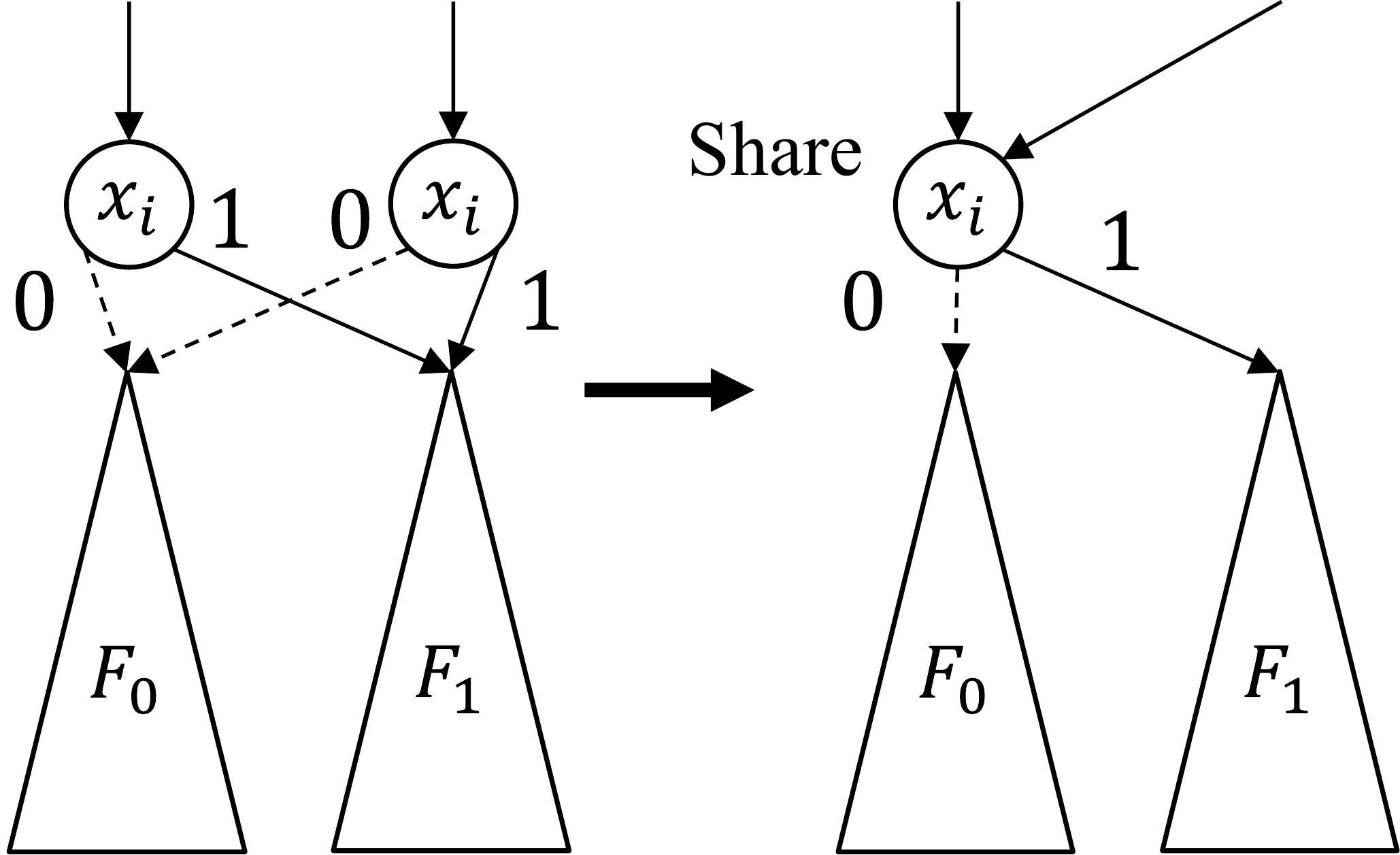}
    \subcaption{}
  \end{minipage}
  \begin{minipage}[b]{0.4\linewidth}
    \centering
    \includegraphics[keepaspectratio, scale=0.14]{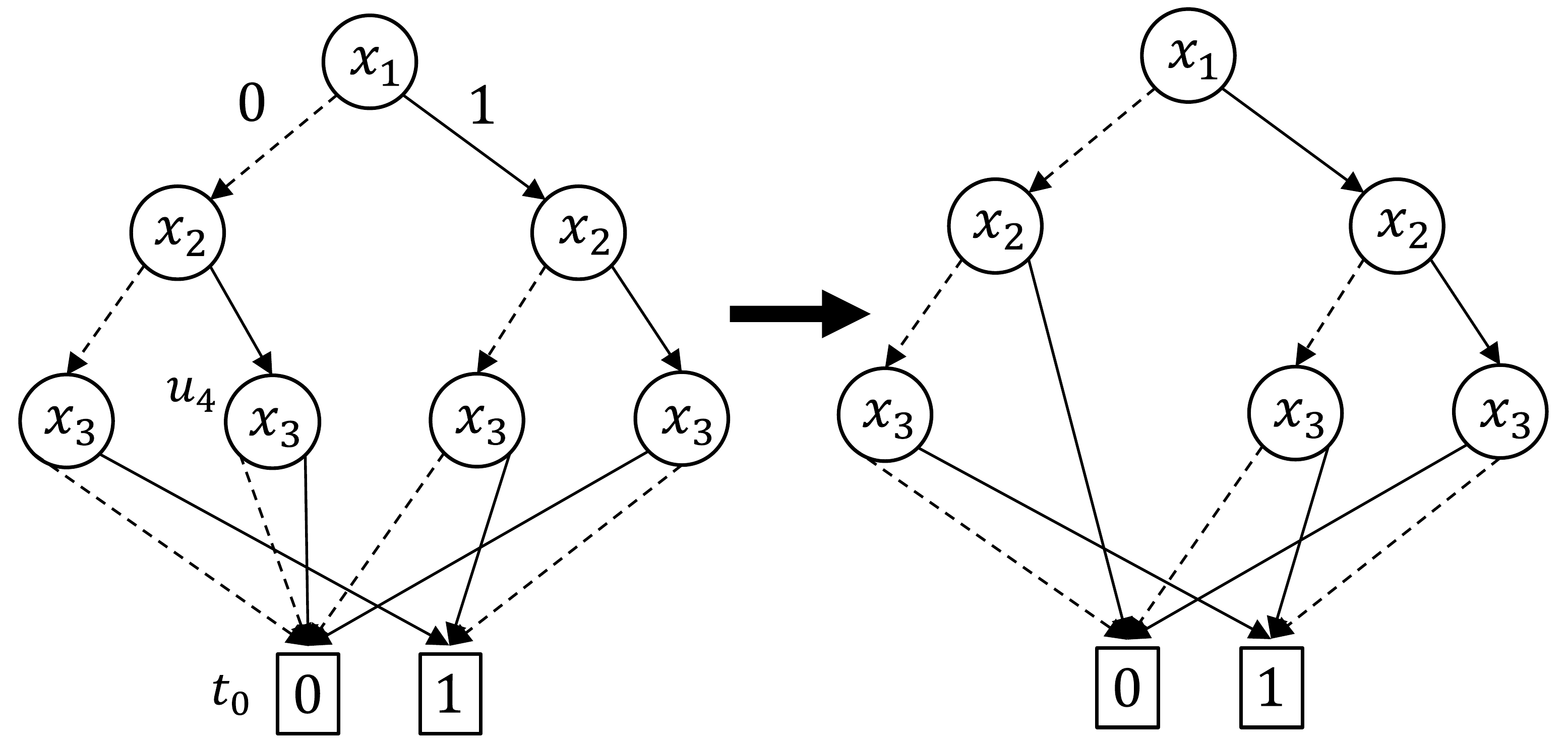}
    \subcaption{}
  \end{minipage}
  \begin{minipage}[b]{0.4\linewidth}
    \centering
    \includegraphics[keepaspectratio, scale=0.14]{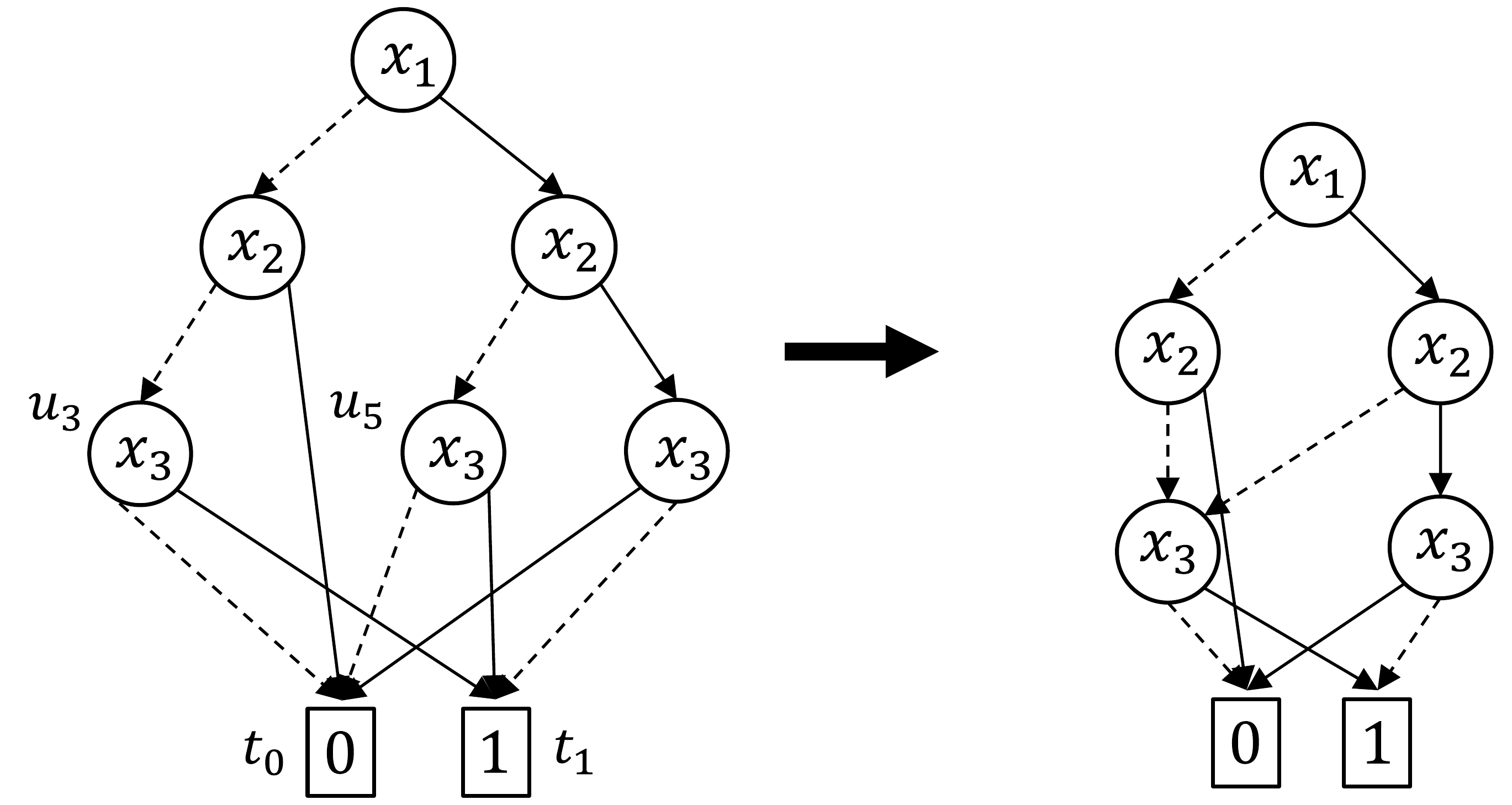}
    \subcaption{}
  \end{minipage}
  \caption{Contraction rules and examples. (a) The redundant node deletion is the deletion of an internal node when two outgoing edges from the internal node have the same child node, where the triangle represents a partial $\mathrm{BDD}_F$ representing a Boolean function $F$ and the top vertex of the triangle denotes the root of the $\mathrm{BDD}_F$. (b) The equivalent node sharing is the merging of two internal nodes when the two internal nodes share both of the child nodes of the two outgoing edges from their respective internal nodes. (c) A $\mathrm{BDD}_f$ representing a Boolean function $f(x_1, x_2, x_3) = \bar{x}_1\bar{x}_2x_3 + x_1\bar{x}_2x_3 + x_1x_2\bar{x}_3$ and the application of the redundant node deletion: The internal node $u_4$ is deleted because the child nodes of the $0$-edge and the $1$-edge are the same $0$-terminal node $t_0$. (d) The application of the equivalent node sharing and the $\mathrm{RBDD}_f$. The internal node $u_5$ is merged to another internal node $u_3$ because both of the child nodes of the two outgoing edges from their respective internal nodes are the $0$-terminal node $t_0$ and the $1$-terminal node $t_1$. } \label{FIG:BDD_rules}
\end{figure}
}

\begin{definition} \label{def:bdd}
A binary decision diagram $\mathrm{BDD}_f = (V,E)$ representing a Boolean function $f(x_1,\ldots, x_n)$ of variables $x_1,\ldots, x_n$ is a rooted directed acyclic graph $G = (V,E)$ with one or two nodes of outdegree zero, where all other nodes are equipped with two labeled outgoing edges and the graph satisfies the following properties.
\begin{itemize}
    \item[1] The nodes of outdegree zero are labeled with $0$ and $1$, respectively, and are called the $0$-terminal node $t_0$ and the $1$-terminal node $t_1$.
    \item[2] The other nodes are labeled with the variables $x_1,\ldots, x_n$, called the internal nodes. There is one internal node of indegree zero, called the root $r$. Let $i(u)$ denote the index of the variable labeling an internal node $u$.
    \item[3] Two outgoing edges of each internal node are labeled with $0$ and $1$, respectively, and are called the $0$-edge and the $1$-edge of the internal node. Let $e_{b}(u)$ and $h_{b}(u)$ denote the $b$-edge of an internal node $u$ and the child node of the $b$-edge. We regard $e_0(u)$ and $e_1(u)$ as labeled edges, so they remain distinct even when they have the same child node. In a general BDD, the two labeled edges may have the same child node, {\it i.e.}, $h_0(u)=h_1(u)$ is allowed.
\end{itemize}
The Boolean contraction rules are the redundant node deletion and the equivalent node sharing. A BDD is called a reduced binary decision diagram $\mathrm{RBDD}_f$ if these contraction rules have been applied exhaustively, equivalently, if the following conditions hold.
\begin{itemize}
    \item[4] The redundant node deletion cannot be applied. For any internal node $u$, two child nodes of the two outgoing edges from $u$ are different, {\it i.e.}, $h_0(u) \neq h_1(u)$.
    \item[5] The equivalent node sharing cannot be applied. Any two internal nodes $u, v$ labeled with the same variable do not share both of the child nodes of two outgoing edges from $u$ and $v$, {\it i.e.}, when $i(u) = i(v)$, $(h_0(u), h_1(u)) \neq (h_0(v), h_1(v))$.
\end{itemize}
\end{definition}

From Ref.~\cite{Breitbar_P1995-dg}, there are several classes of binary decision diagrams (BDDs). 
In this study, to use the results of Ref.~\cite{Breitbar_P1995-dg}, we introduce the following free BDDs and ordered BDDs. (See Fig.~\ref{fig:ex_obdd_fbdd_wfbdd}.)

{\nolinenumbers
\begin{figure}[t]
    \begin{minipage}[b]{0.4\linewidth}
        \centering
        \includegraphics[keepaspectratio, scale=0.17]{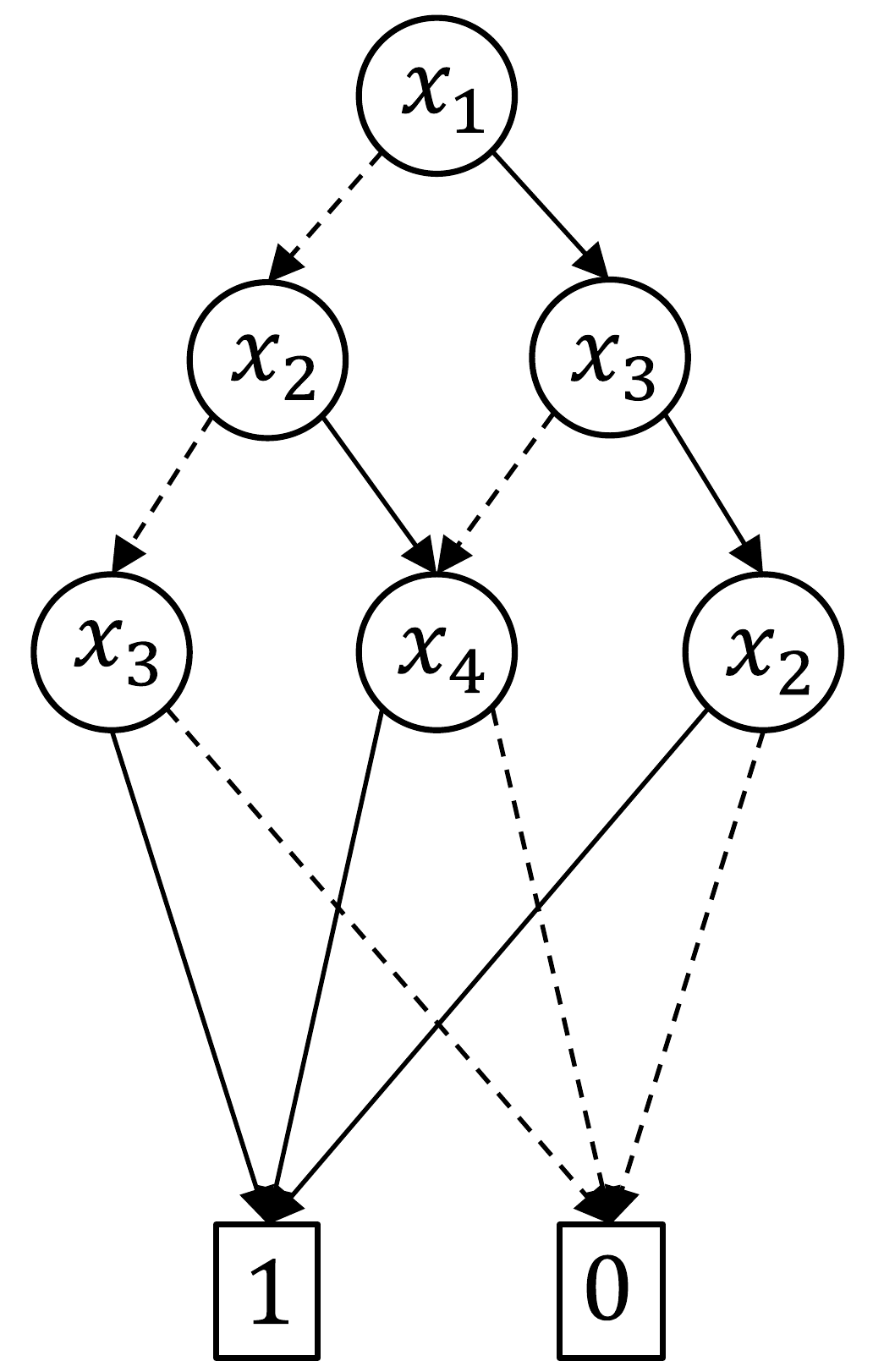}
        \subcaption{}
    \end{minipage}
    \begin{minipage}[b]{0.4\linewidth}
        \centering
        \includegraphics[keepaspectratio, scale=0.17]{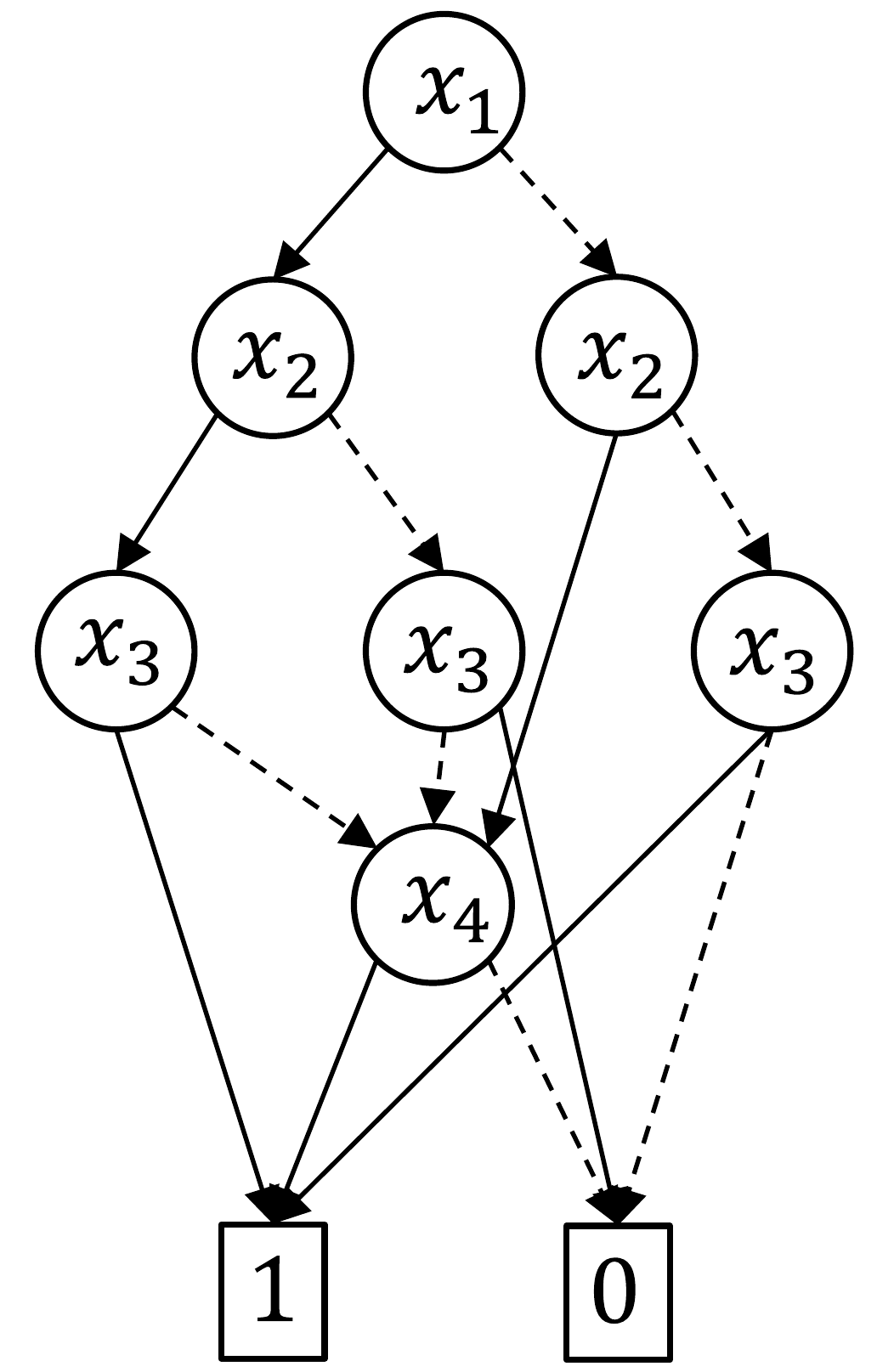}
        \subcaption{}
    \end{minipage}
    \begin{minipage}[b]{0.4\linewidth}
        \centering
        \includegraphics[keepaspectratio, scale=0.17]{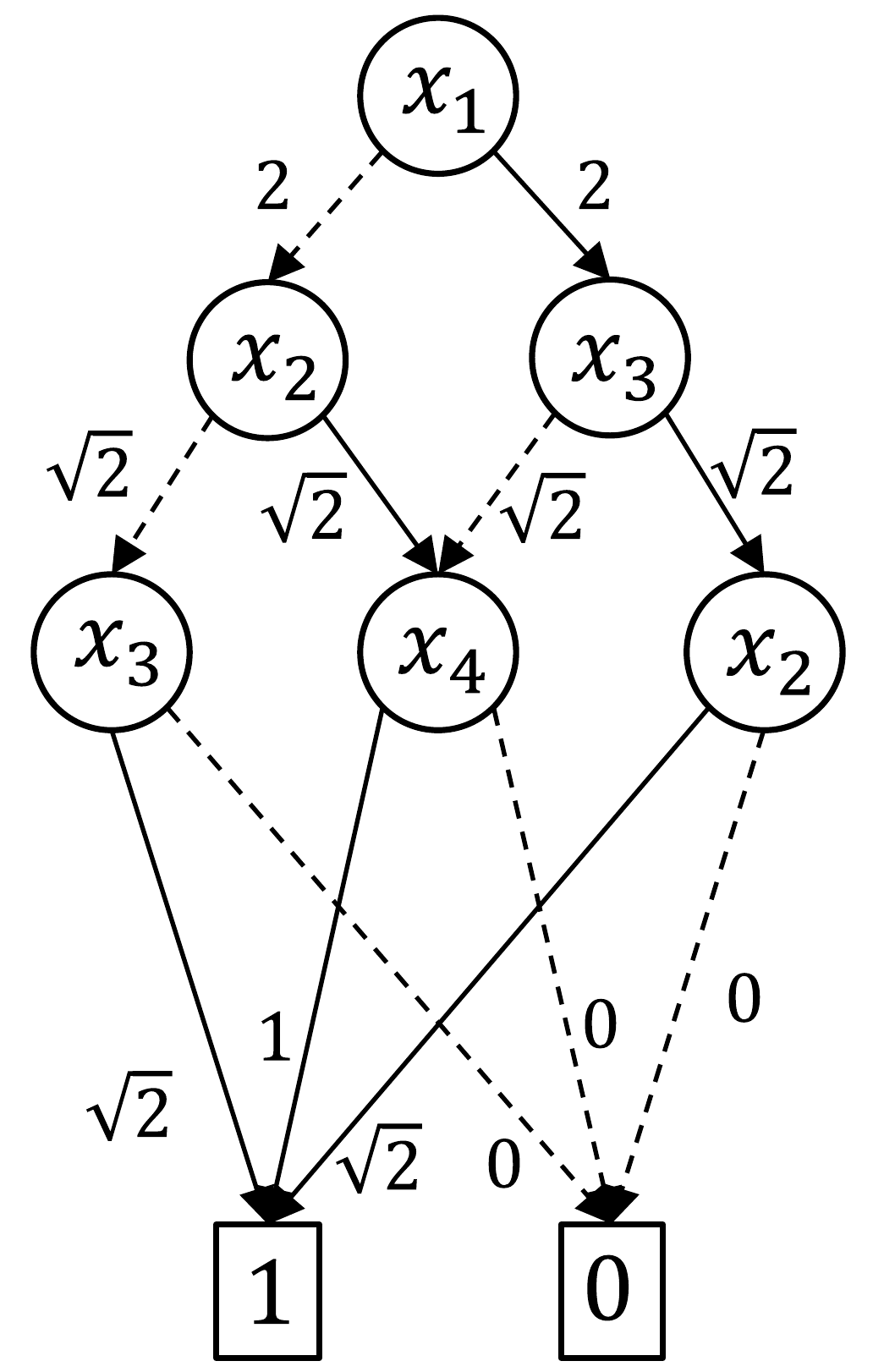}
        \subcaption{}
    \end{minipage}
    \caption{FBDD, OBDD, and WFBDD for $f = \bar{x}_1 \bar{x}_2 x_3 + \bar{x}_1 x_2 x_4 + x_1 \bar{x}_3 x_4 + x_1 x_2 x_3$. (a) In the FBDD, each path from the root to the terminal nodes has at most one occurrence of each variable. (b) In the OBDD, the order of occurrences of the variables on each path from the root to the terminal nodes is consistent. (c) In the WFBDD, directed edges are assigned complex weights. Under the weighted path semantics below, an input reaching $t_0$ has weighted value zero, whereas an input reaching $t_1$ has the product of the selected edge weights.} \label{fig:ex_obdd_fbdd_wfbdd}
\end{figure}
}

\begin{definition} \label{def:fbdd}
A free binary decision diagram $\mathrm{FBDD}_f$ representing a Boolean function $f(x_1,\ldots, x_n)$ of variables $x_1,\ldots, x_n$ is a $\mathrm{BDD}_f$ for which each path from the root to the terminal nodes has at most one occurrence of each variable. If, in addition, the underlying BDD is reduced in the sense of Def.~\ref{def:bdd}, we call it a reduced FBDD.
\end{definition}

\begin{definition}
An ordered binary decision diagram $\mathrm{OBDD}_f$ representing a Boolean function $f(x_1,\ldots, x_n)$ of variables $x_1,\ldots, x_n$ is an $\mathrm{FBDD}_f$ such that the order of occurrences of the variables on each path from the root to the terminal nodes is consistent with some linear order. If, in addition, the underlying BDD is reduced in the sense of Def.~\ref{def:bdd}, we call it a reduced OBDD.
\end{definition}

The FBDD and OBDD have the following Boolean path semantics. For an input $x\in\{0,1\}^n$, starting from the root, we follow the $b$-edge of a node labeled by $x_i$ when $x_i=b$. Since an FBDD queries each variable at most once along each root-to-terminal path, this procedure determines a unique terminal node. The Boolean value is $0$ or $1$ according to whether the reached terminal is $t_0$ or $t_1$, respectively. Variables that are not queried on the path are treated as free variables.

We use WFBDDs here first as classical weighted decision diagrams. A WFBDD augments the Boolean path semantics of an FBDD by complex edge weights, while the quantum-state semantics used later for state preparation is deferred to Sec.~\ref{Sec3}. To establish this data structure more formally, we define the FBDD with weighted edges, called the weighted FBDD. (See Fig.~\ref{fig:ex_obdd_fbdd_wfbdd}.)

\begin{definition} \label{def:wbdd}
A weighted free binary decision diagram $\mathrm{WFBDD}_f=(V,E)$ is an $\mathrm{FBDD}_f$ representing a Boolean function $f(x_1,\ldots,x_n)$ of variables $x_1,\ldots,x_n$ equipped with a weight map $w:E\to\mathbb{C}$. For an internal node $u$, let $w_b(u):=w(e_b(u))$ denote the weight of the $b$-edge of $u$. A WFBDD is not required to be reduced in the Boolean sense; in particular, two labeled outgoing edges may have the same child node while carrying different weights.

\end{definition}

The weight map is a decoration of the underlying FBDD. The root-to-terminal paths are those of the underlying FBDD, and no local normalization condition is imposed in Def.~\ref{def:wbdd}. Boolean contraction rules are not imposed after weights are assigned, because two edges with the same child node may carry different weights and therefore may represent different amplitude factors.
Note that a weighted $\mathrm{OBDD}$ ($\mathrm{WOBDD}$) is obtained analogously by assigning complex edge weights to an OBDD; its Boolean path semantics is inherited from the underlying OBDD, and its weighted path semantics is defined as below.

A WFBDD has the following classical weighted path semantics. For an input $z\in\{0,1\}^n$, let $P_G(z)=(V_z,E_z)$ denote the unique root-to-terminal path obtained by starting from the root and, at each internal node $u$ labeled by $x_{i(u)}$, following the $z_{i(u)}$-edge. When the WFBDD $G$ is clear from the context, we simply write $P(z)$. If this path reaches $t_0$, the weighted value is defined to be zero. If it reaches $t_1$, the weighted value is defined to be the product of the edge weights selected along the path. Thus a WFBDD defines a pseudo-Boolean function $F_G:\{0,1\}^n\to\mathbb{C}$ by
\begin{eqnarray}
F_G(z)
=
\begin{cases}
0, & \text{if } P_G(z) \text{ reaches } t_0,\\
\displaystyle \prod_{u\in V_z\setminus\{t_0,t_1\}} w_{z_{i(u)}}(u), & \text{if } P_G(z) \text{ reaches } t_1,
\end{cases}
\label{eq:weighted_path_semantics}
\end{eqnarray}
This is a classical weighted path semantics of the data structure and is distinct from the quantum-state semantics introduced in Sec.~\ref{Sec3}.

For a WFBDD $G$, let $V_{\mathrm{rel}}(G)$ denote the set of internal nodes lying on at least one root-to-$t_1$ path. An internal node in $V_{\mathrm{rel}}(G)$ is called relevant. For an internal node $u$, the ordered pair $(w_0(u),w_1(u))$ is called its outgoing weight pair. A nonzero outgoing weight pair is called locally normalized if
\begin{eqnarray}
    |w_0(u)|^2+|w_1(u)|^2=1.
    \label{eq:local_normalization_condition}
\end{eqnarray}

\begin{definition}[QSP-admissible WFBDD]
\label{def:qsp_admissible_wfbdd}
A WFBDD $G=\mathrm{WFBDD}_f$ is QSP-admissible if it satisfies the following conditions.
\begin{itemize}
    \item[1] The underlying FBDD has at least one accepting input.
    \item[2] For every $u\in V_{\mathrm{rel}}(G)$, the outgoing weight pair $(w_0(u),w_1(u))$ is locally normalized.
    \item[3] For every $u\in V_{\mathrm{rel}}(G)$ and $b\in\{0,1\}$, if $e_b(u)$ lies on no root-to-$t_1$ path, then $w_b(u)=0$.
\end{itemize}

\end{definition}

Given a WFBDD $G$ in Def.~\ref{def:wbdd} whose relevant internal nodes have nonzero outgoing weight pairs, its locally normalized representative $\widetilde{G}$ is obtained by replacing, for every $u\in V_{\mathrm{rel}}(G)$,
\begin{eqnarray}
    w_b(u)
    \quad \text{with} \quad
    \widetilde{w}_b(u)
    =
    \frac{w_b(u)}{\sqrt{|w_0(u)|^2+|w_1(u)|^2}},
    \qquad b\in\{0,1\}.
    \label{eq:local_normalized_weight}
\end{eqnarray}
The underlying FBDD, its Boolean path semantics, and its classical weighted path support are unchanged. Equation~(\ref{eq:local_normalized_weight}) enforces Condition~2 of Def.~\ref{def:qsp_admissible_wfbdd} for the representative but does not by itself enforce Condition~3. When $\widetilde{G}$ satisfies Def.~\ref{def:qsp_admissible_wfbdd}, the general WFBDD $G$ is used below through this locally normalized representative, which is QSP-admissible.

\section{Quantum state preparation via weighted FBDD} \label{Sec3}

In this section, we show a state-preparation procedure for QSP via WFBDDs.
All quantum circuits in this section are constructed from single-qubit unitaries, controlled single-qubit unitaries, and Toffoli gates.
As explained in Sec.~\ref{Sec2}, a general WFBDD whose locally normalized representative in Eq.~(\ref{eq:local_normalized_weight}) is QSP-admissible is used through that representative. Thus, throughout this section, we write $G=\mathrm{WFBDD}_f$ for a QSP-admissible WFBDD unless otherwise stated. In Sec.~\ref{sec:description}, we define the quantum-state semantics of a QSP-admissible WFBDD. The underlying FBDD determines which inputs are rejected, while the QSP-admissible outgoing weights determine the amplitudes on accepting paths.
Then, we give an example of an $O(\mathrm{poly}(n))$-sized $\mathrm{FBDD}_f$ whose weighted version is interpreted as a quantum state through Def.~\ref{def:q_state_wfbdd}.
In Sec.~\ref{sec:quantum_algorithm}, we show that Algorithm~\ref{alg:WFBDD_SG} prepares the state defined by Def.~\ref{def:q_state_wfbdd} from a QSP-admissible WFBDD.

\subsection{Classical description of quantum states by weighted FBDD}
\label{sec:description}

In this subsection, we define how a QSP-admissible WFBDD is read as a quantum state. The definition does not use the classical weighted path value $F_G$ directly; rather, it uses the QSP-admissible edge weights and assigns zero amplitude to inputs whose selected path reaches $t_0$.
First, we define the quantum-state semantics as follows.

\begin{definition}[Quantum-state semantics of a QSP-admissible WFBDD]
\label{def:q_state_wfbdd}
Let $G=\mathrm{WFBDD}_f$ be a QSP-admissible WFBDD in Def.~\ref{def:qsp_admissible_wfbdd}. For $z \in \{0,1\}^n$, let $P(z)=(V_z,E_z)$ be the unique root-to-terminal path obtained by starting from the root and, at each internal node $u$ labeled by $x_{i(u)}$, following the $z_{i(u)}$-edge until reaching one of the terminal nodes. The quantum state associated with $G$ is defined as
    \begin{eqnarray}
        \ket{\psi} = \sum_{z \in \{ 0,1 \}^n} \alpha(z) \ket{z}, \label{eq:wfbdd_state}
    \end{eqnarray}
    where the amplitude $\alpha(z)$ is given piecewise by
    \begin{eqnarray}
        \alpha(z) =
        \begin{cases}
        0, & \text{if } P(z) \text{ reaches } t_0,\\[2mm]
        \displaystyle
        \frac{1}{\sqrt{2^{n-|V_z|+1}}}
        \prod_{u \in V_z \setminus \{ t_0, t_1 \}}
        w_{z_{i(u)}}(u),
        & \text{if } P(z) \text{ reaches } t_1.
        \end{cases}
        \label{eq:alpha_q_state_piecewise}
    \end{eqnarray}

\end{definition}

If one starts from a general WFBDD in Def.~\ref{def:wbdd}, the associated state is defined by first passing to its locally normalized representative in Eq.~(\ref{eq:local_normalized_weight}), provided that this representative is QSP-admissible. For accepting paths in a QSP-admissible WFBDD, Eq.~(\ref{eq:alpha_q_state_piecewise}) is the product of the QSP-admissible branch weights and the uniform amplitudes associated with variables not queried on the path.

The prefactor $2^{-(n-|V_z|+1)/2}$ is not part of the local edge-weight normalization. Since $|V_z|-1$ is the number of internal nodes and hence the number of variables queried along $P(z)$, the exponent $n-|V_z|+1$ is the number of variables not queried on that path. These unqueried variables remain in the uniform state $\ket{+}$ and contribute a factor $2^{-1/2}$ per such variable. In general, this prefactor cannot be absorbed into a single root weight, because the number of unqueried variables may depend on the selected root-to-terminal path.

The piecewise form in Eq.~(\ref{eq:alpha_q_state_piecewise}) separates the accepting-path product from the rejecting-path convention: if $P(z)$ reaches $t_0$, the amplitude is set to zero before any accepting-path product is evaluated. The accepting-support condition in Def.~\ref{def:qsp_admissible_wfbdd} ensures that this zero-amplitude convention is consistent with the local unitary construction.

\begin{lemma} \label{lem:1}
    Let $G$ be a QSP-admissible WFBDD and let $\alpha(z)$ be defined by Def.~\ref{def:q_state_wfbdd}. Then $|\alpha(z)|^2$ is a probability distribution, {\it i.e.}, $|\alpha(z)|^2 \in [0,1]$ for $z \in \{ 0,1 \}^n$ and $\sum_z |\alpha(z)|^2 = 1$.
\end{lemma}

\begin{proof}[Proof sketch]
    Starting from unit probability mass at the root, Condition~2 of Def.~\ref{def:qsp_admissible_wfbdd} locally preserves probability mass at each relevant internal node. Condition~3 of Def.~\ref{def:qsp_admissible_wfbdd} prevents nonzero mass from flowing into branches outside the accepting support. For a fixed accepting root-to-$t_1$ path, the variables not queried on the path remain in the state $\ket{+}$; summing over all assignments of these free variables cancels the free-variable prefactor in Eq.~(\ref{eq:alpha_q_state_piecewise}) and leaves exactly the squared product of the branch weights along the path. Summing these contributions over all accepting paths gives the unit probability mass initially placed at the root. A detailed normalization argument is given in Appendix~\ref{apdx:proof_lemma1}.
\end{proof}

This semantics is a structured WFBDD-based amplitude model: rejecting inputs are determined by the underlying FBDD, and the accepting amplitudes are determined by the QSP-admissible edge data together with the free-variable prefactor. It is therefore distinct from a model in which arbitrary amplitudes, independent of the diagram structure, are supplied by an external amplitude oracle or by QRAM-style data access.

Following Ref.~\cite{Breitbar_P1995-dg}, we give an example of an $O(\mathrm{poly}(n))$-sized FBDD whose weighted version gives a compact WFBDD description under the quantum-state semantics of Def.~\ref{def:q_state_wfbdd}.
First, let us define the following symmetric Boolean function.
\begin{definition}
Let $X = (x_1, \ldots, x_n)$. The elementary symmetric Boolean function $S^i_n(X)$ is equal to one if and only if exactly $i$ variables in $X$ are equal to one, where $0 \le i \le n$.
\end{definition}

Second, we define the following function (see Fig.~\ref{fig:sym_function_and_G}),
\begin{eqnarray}
f^1 (X, Y, Z) &=& \sum_{i=1}^n x_i (S^{i-1}_{2n}(Y,Z) + S^{n+i}_{2n}(Y,Z)), \\
f^2 (X, Y, Z) &=& f^1 (Y, Z, X), \\
f^3 (X, Y, Z) &=& f^1 (Z, X, Y).
\end{eqnarray}
Then, we define the following Boolean function $h$ (see Fig.~\ref{fig:sym_function_and_G}),
\begin{eqnarray}
h(X,Y,Z) = \bar{v}\bar{w} f^1(X,Y,Z) + \bar{v} w f^2(X,Y,Z) + v\bar{w} f^3(X,Y,Z). \label{eq:dense_h}
\end{eqnarray}

The following result (Theorem 7 of Ref.~\cite{Breitbar_P1995-dg}) is known. For every $n \ge 4$, the $\mathrm{FBDD}_h$ is computable with no more than $3(2n^2+2n)+3$ nodes; however, every OBDD computing $h$ needs at least $2^{n/3}$ nodes. Let $m=3n+2$ be the total number of input variables of $h$. If $h$ has $d$ satisfying assignments, then one can construct an OBDD of size $O(md)$ by listing the satisfying assignments. Since every OBDD computing $h$ needs at least $2^{n/3}$ nodes, we have
\[
    d \ge 2^{n/3}/O(m)=2^{\Omega(n)}=2^{\Omega(m)}.
\]
Therefore, after equipping this FBDD with edge weights so as to obtain a QSP-admissible WFBDD and interpreting it through Def.~\ref{def:q_state_wfbdd}, one obtains a compact WFBDD description of an $m$-qubit quantum state with $d=2^{\Omega(m)}$ nonzero coefficients.

{\nolinenumbers
\begin{figure}[t]
    \begin{minipage}[b]{0.45\linewidth}
        \centering
        \includegraphics[keepaspectratio, scale=0.11]{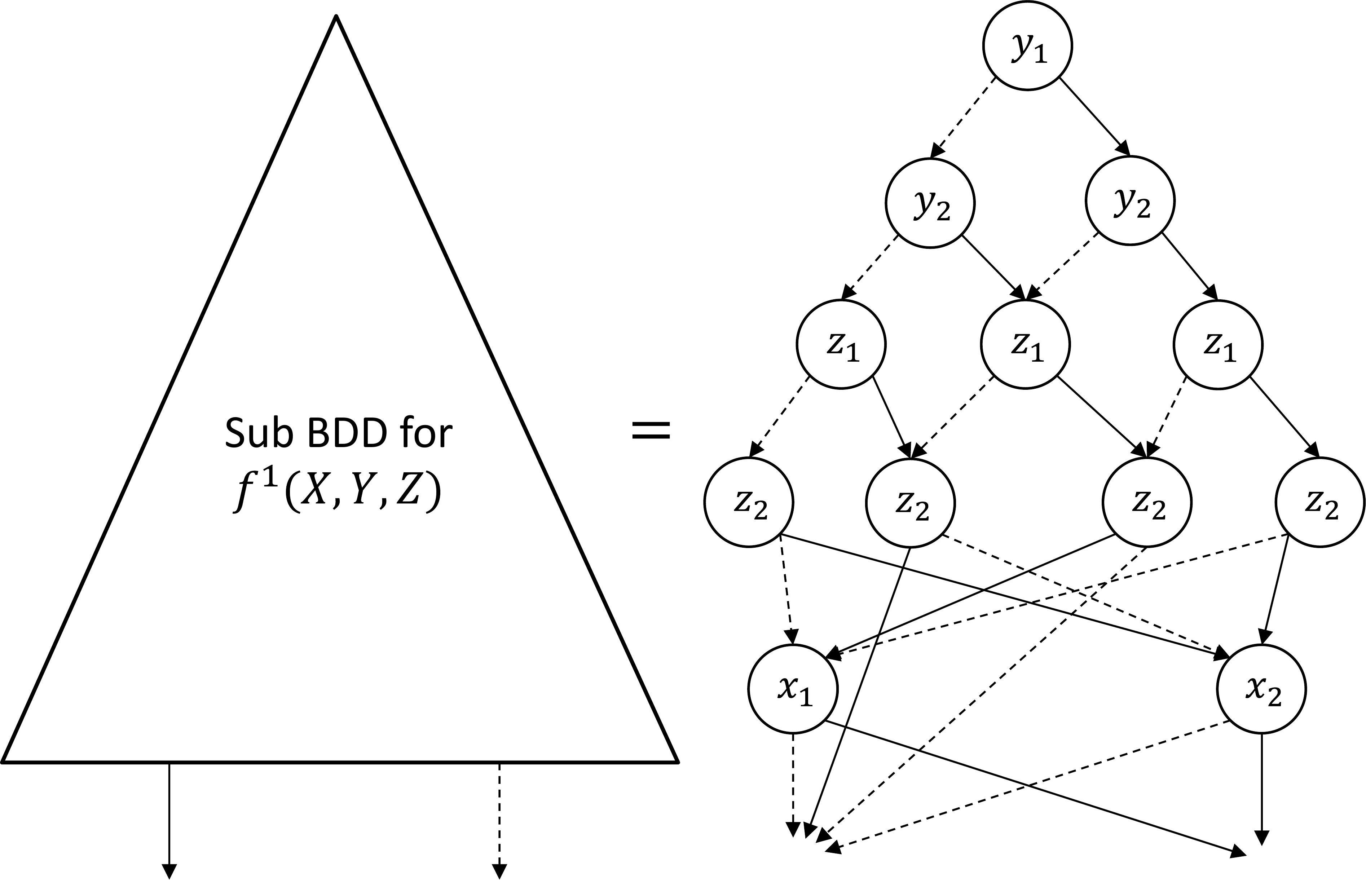}
        \subcaption{}
    \end{minipage}
    \begin{minipage}[b]{0.45\linewidth}
        \centering
        \includegraphics[keepaspectratio, scale=0.15]{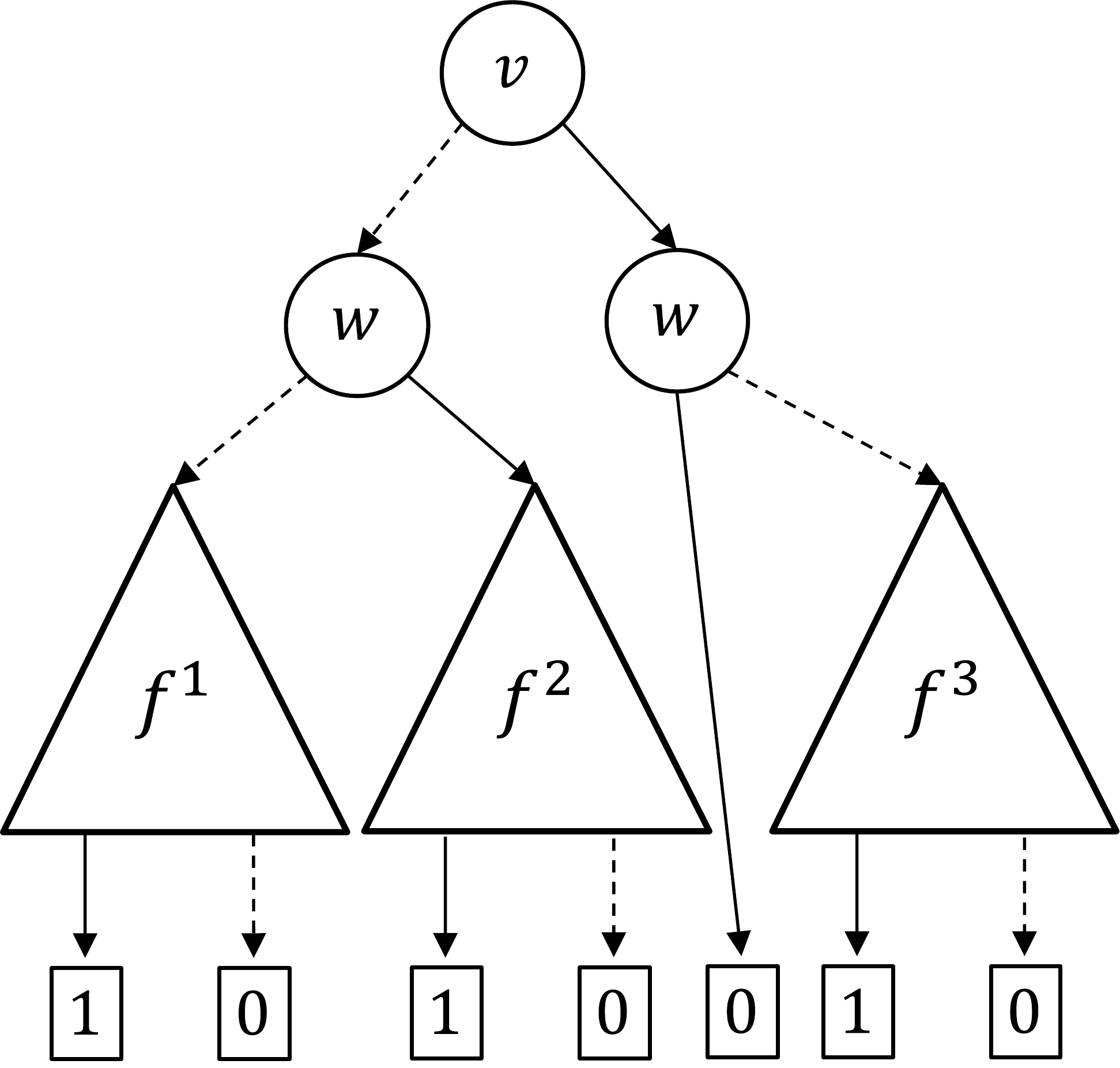}
        \subcaption{}
    \end{minipage} 
    \caption{(a) A $6$-variable example of subBDD for $f^1(X,Y,Z) = x_1(S_4^0(Y,Z)+S_4^3(Y,Z)) + x_2(S_4^1(Y,Z)+S_4^4(Y,Z))$, where $X=(x_1, x_2), Y=(y_1, y_2), Z=(z_1,z_2)$. The function $S_n^i(x_1,\ldots, x_n)$ is called the elementary symmetric function; it is equal to one if and only if exactly $i$ of its variables are equal to one, where $0 \le i \le n$. (b) A $\mathrm{FBDD}_h$ for $h(X,Y,Z) = \bar{v}\bar{w} f^1(X,Y,Z) + \bar{v} w f^2(X,Y,Z) + v\bar{w} f^3(X,Y,Z)$.} \label{fig:sym_function_and_G}
\end{figure}
}

\subsection{\texorpdfstring{State-preparation circuit from weighted FBDD}{State-preparation circuit from weighted FBDD}}
\label{sec:quantum_algorithm}

Now, we show that we can construct an efficient quantum circuit generating the quantum state defined from a given QSP-admissible $\mathrm{WFBDD}_f$ by Def.~\ref{def:q_state_wfbdd} (see an example in Fig.~\ref{example:alg:SG}).
In this section, the WFBDD is treated as a supplied classical description of the target state. The gate count below counts the generated state-preparation circuit and does not include the cost of constructing the WFBDD description or any coherent data-access mechanism.
If the supplied data are a general WFBDD in Def.~\ref{def:wbdd} and if its locally normalized representative in Eq.~(\ref{eq:local_normalized_weight}) is QSP-admissible, then Algorithm~\ref{alg:WFBDD_SG} is applied to that representative. Since this representative has the same underlying FBDD, the node count and the asymptotic gate count are unchanged.

\begin{figure}[!t]
\begin{algorithm}[H]
{\nolinenumbers
\caption{StateGeneration} \label{alg:WFBDD_SG}
}
\begin{algorithmic}[1]
\Require A QSP-admissible $\mathrm{WFBDD}_f = (V,E)$ in Def.~\ref{def:qsp_admissible_wfbdd}. Here $i(u)$, $e_b(u)$, $h_b(u)$, and $w_b(u)$ denote the index of the variable assigned to an internal node $u$, the $b$-edge out from $u$, the child node of $e_b(u)$, and the weight assigned to $e_b(u)$, respectively.
\Ensure The quantum state $\sum_{z \in \{ 0,1 \}^n} \alpha(z) \ket{z}$ associated with $\mathrm{WFBDD}_f$ by Def.~\ref{def:q_state_wfbdd}. 
\State {\bf [Initialization]} \label{alg:SG1}
\State Prepare a set $Q = \{ r \}$ and an empty stack $S$. \label{alg:SG2}
\State Assign qubits initialized to $\ket{+}$ to the variables $x_1, \ldots, x_n$. \label{alg:SG3}
\State Assign a qubit initialized to $\ket{1}$ to the root $r$. \label{alg:SG4}
\State Assign qubits initialized to $\ket{0}$ to the internal nodes except the root $r$. \label{alg:SG5}
\State {\bf [Main loop]} \label{alg:SG6}
\While{the set $Q$ is not empty} \label{alg:SG7}
    \State Remove a node $u$ from $Q$. \label{alg:SG8}
    \State Push the node $u$ to the stack $S$. \label{alg:SG9}
    \State Apply the controlled-$H$ gate controlled by the qubit $u$ to the qubit $x_{i(u)}$. \label{alg:SG10}
    \State Apply the controlled-$U(u)$ gate controlled by the qubit $u$ to the qubit $x_{i(u)}$. If $u$ is relevant, use
    $U(u) =
    \begin{pmatrix}
        w_0(u) & -w^{\ast}_1(u) \\
        w_1(u) & w^{\ast}_0(u) \\
    \end{pmatrix}$,
    whose first column is $(w_0(u),w_1(u))^T$; if $u$ is irrelevant, choose any fixed single-qubit unitary, e.g., the identity. \label{alg:SG11}
    \ForEach{$b$-edge $e_b(u)$ out from the node $u$ and its child node $h_b(u)$} \label{alg:SG12}
        \State Remove the $b$-edge $e_b(u)$ from $E$. \label{alg:SG13}
        \If{the child node $h_b(u)$ has no incoming edge and is not a terminal node} \label{alg:SG14}
            \State Add the child node $h_b(u)$ to the set $Q$. \label{alg:SG15}
        \EndIf \label{alg:SG16}
        \If{$h_b(u) \neq t_0, t_1$} \label{alg:SG17}
            \State Apply the $X$ gate to the qubit $x_{i(u)}$ when $b = 0$. \label{alg:SG18}
            \State Apply the Toffoli gate controlled by the qubits $u$ and $x_{i(u)}$ to the qubit $h_b(u)$. \label{alg:SG19}
            \State Apply the $X$ gate to the register assigned to $x_{i(u)}$ when $b = 0$. \label{alg:SG20}
        \EndIf \label{alg:SG21}
    \EndFor \label{alg:SG22}
\EndWhile \label{alg:SG23}
\State {\bf [Undo]} \label{alg:SG24}
\While{the stack $S$ is not empty} \label{alg:SG25}
    \State Pop a node $u$ from the stack $S$. \label{alg:SG26}
    \State For each original $b$-edge $e_b(u)$ out from $u$, repeat Lines~\ref{alg:SG17}--\ref{alg:SG21}.  \label{alg:SG27}
\EndWhile \label{alg:SG28}
\end{algorithmic}
\end{algorithm}
\end{figure}

For every relevant internal node of a QSP-admissible WFBDD, $|w_0(u)|^2+|w_1(u)|^2=1$ and hence the matrix $U(u)$ used in Line~\ref{alg:SG11} of Algorithm~\ref{alg:WFBDD_SG} is unitary. If one starts from raw weights of a general WFBDD, the weights used in that line are those of the locally normalized representative in Eq.~(\ref{eq:local_normalized_weight}), provided that this representative is QSP-admissible. Irrelevant internal nodes do not enter any accepting-path product in Def.~\ref{def:q_state_wfbdd}, so the choice of $U(u)$ at such nodes does not affect the prepared state.

\begin{theorem} \label{thm:WS}
    Let $\mathrm{WFBDD}_f=(V,E)$ be a QSP-admissible WFBDD. Then Algorithm~\ref{alg:WFBDD_SG} constructs a quantum circuit preparing the quantum state associated with $\mathrm{WFBDD}_f$ by Def.~\ref{def:q_state_wfbdd}, with $O(|V|)$ uses of single- or two-qubit gates and $|V|-2$ ancillary qubits. The same statement applies to a general WFBDD whose locally normalized representative is QSP-admissible, by applying the theorem to that representative.
\end{theorem}

{\nolinenumbers
\begin{figure}[t]
    \begin{minipage}[b]{0.4\linewidth}
        \centering
        \includegraphics[keepaspectratio, scale=0.07]{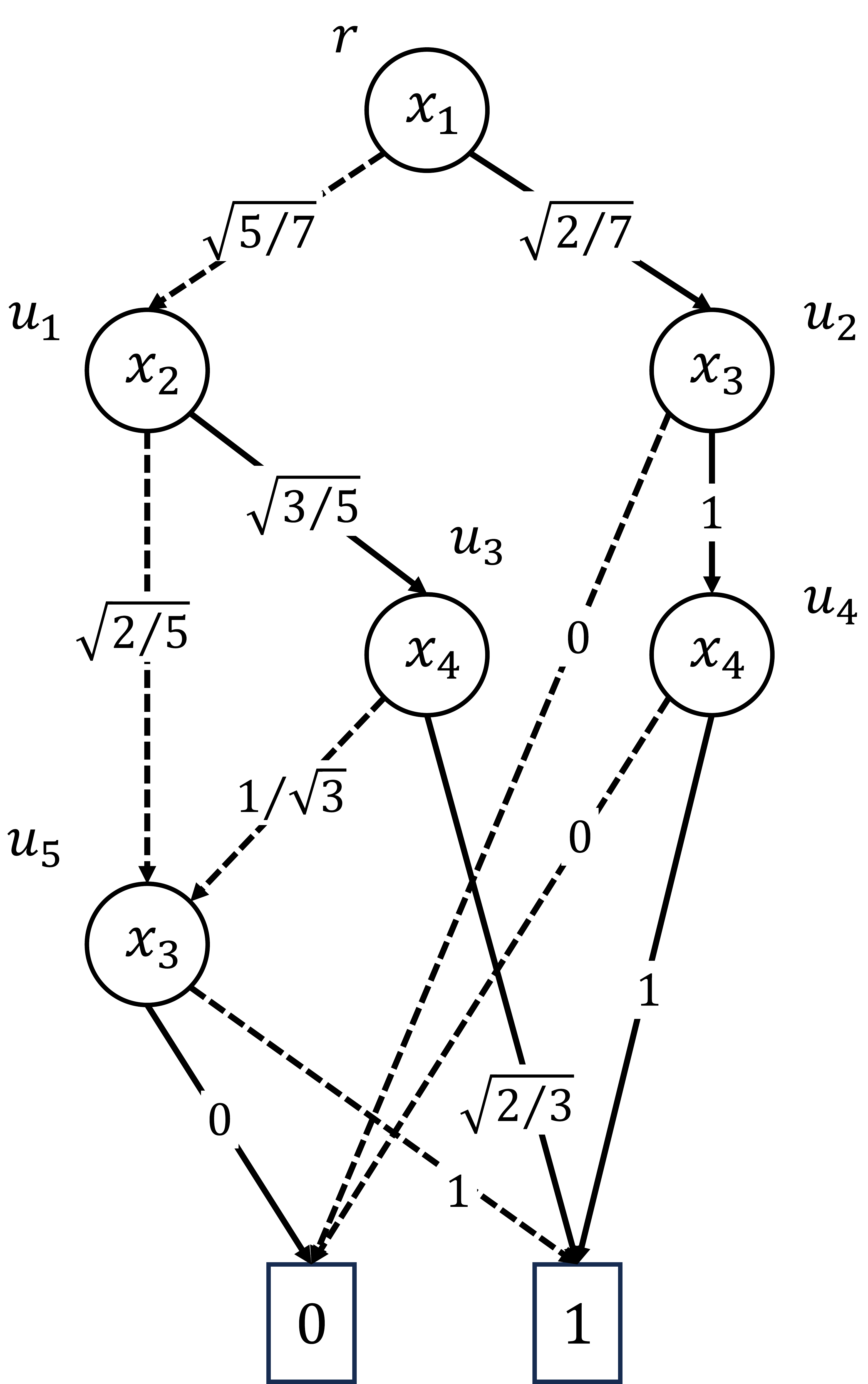}
        \subcaption{}
    \end{minipage}
    \begin{minipage}[b]{0.5\linewidth}
        \centering
        \includegraphics[keepaspectratio, scale=0.15]{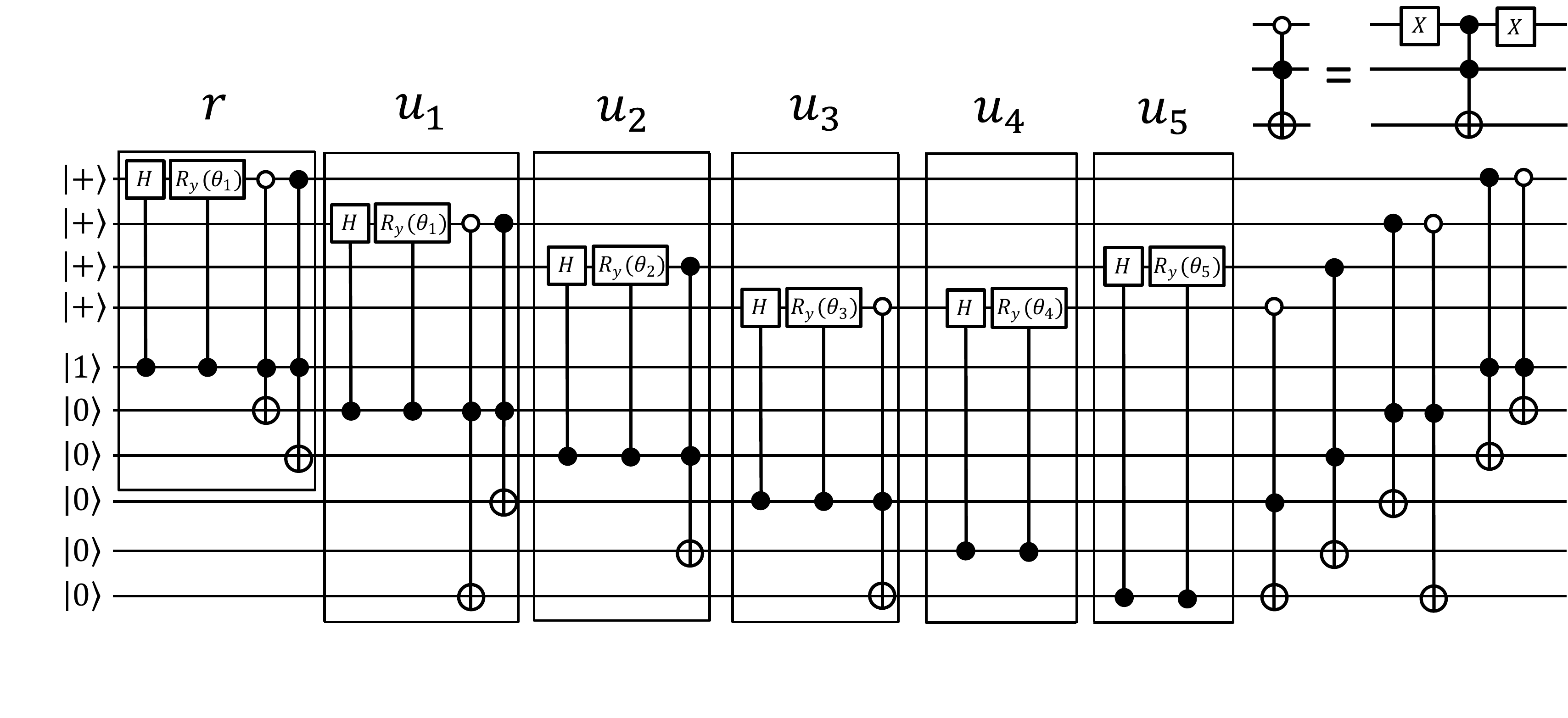}
        \subcaption{}
    \end{minipage}
    \caption{An example of Algorithm~\ref{alg:WFBDD_SG}. (a) An example of $\textrm{WFBDD}_f$. In Algorithm~\ref{alg:WFBDD_SG}, all the internal nodes are pushed to a stack $S$ in a topological ordering such that for every directed edge $(u, v)$ out from $u$ to $v$, $u$ comes before $v$, {\it i.e.}, $r \prec u_1 \prec \cdots \prec u_5$. Start from the root $r$ and push $r$ to $S$. Then, remove the $0$-edge $(r, u_1)$ and the $1$-edge $(r, u_2)$, at which, if the indegree of $u_{1(2)}$ is zero, $u_{1(2)}$ is added to a set $Q$. The added nodes are pushed to $S$ in the following rounds and are ordered. Repeating this process, we push all the internal nodes to $S$ and order them. (b) The quantum circuit generated from the $\textrm{WFBDD}_f$. Initially, assign qubits to the variables $x_1,\ldots, x_4$ and the internal nodes $r, u_1, \ldots, u_5$. Each qubit $x_i$ is set to $\ket{+}$, the qubit $r$ is set to $\ket{1}$, and each qubit $u_j$ is set to $\ket{0}$. Then, the controlled-$H$ gate and the controlled-$U(r)$ gate (as in Line~\ref{alg:SG11}) controlled by the qubit $r$ are applied to the qubit $x_1$. Furthermore, the Toffoli gate controlled by the qubits $r$ and $x_1$ is applied to each qubit $u_{1(2)}$ as shown in the figure. Preparing the unitaries for the rest of the internal nodes in the topological order, we derive the quantum circuit in the figure.} \label{example:alg:SG}
\end{figure}
}

\begin{proof}[Proof sketch]
    Algorithm~\ref{alg:WFBDD_SG} first processes the internal nodes in a topological order obtained by the queue-and-edge-removal procedure. In this procedure, a child node is added to the queue only after all incoming edges from unprocessed internal nodes have been removed. Hence every parent is pushed to the stack before its internal children, and, because the WFBDD is a finite directed acyclic graph, the procedure eventually orders all internal nodes topologically. Along this order, one proves by induction that the variable register accumulates the branch-weight product and the free-variable prefactor of Def.~\ref{def:q_state_wfbdd}, while the node registers record the internal nodes activated along the selected partial paths. The controlled operation in Line~\ref{alg:SG11} of Algorithm~\ref{alg:WFBDD_SG} inserts the QSP-admissible branch pair $(w_0(u),w_1(u))^T$, and the Toffoli gates propagate the active-path information to the child-node registers. By the accepting-support condition, branches outside all root-to-$t_1$ paths have zero amplitude. The undo stage applies the child-recording Toffoli gates in reverse topological order, returning the node registers to their initial product state, namely $\ket{1}$ on the root register and $\ket{0}$ on the other internal-node registers, without changing the variable-register amplitudes. Thus the remaining variable register is the state of Def.~\ref{def:q_state_wfbdd}. The number of controlled one-qubit operations and Toffoli gates is linear in the number of nodes, and a constant-size Toffoli decomposition gives the stated single- and two-qubit gate count. The detailed proof is given in Appendix~\ref{apdx:proof_theorem1}.
\end{proof}

\section{Quantum state preparation for uniform superpositions} \label{Sec4}

In the previous section, we defined the quantum-state semantics of a WFBDD through its locally normalized representative and gave Algorithm~\ref{alg:WFBDD_SG} for preparing the resulting state. In this section, we specialize this framework to the uniform superposition over the accepting assignments of a Boolean function.

Let $f:\{0,1\}^n\to\{0,1\}$ be represented by a supplied $\mathrm{FBDD}_f=(V,E)$. We assume throughout this section that $|f|>0$, because otherwise the uniform state over satisfying assignments is not defined. The target state is
\begin{eqnarray}
    \ket{\phi_f}
    =
    \frac{1}{\sqrt{|f|}}
    \sum_{z:f(z)=1}\ket{z}
    =
    \frac{1}{\sqrt{|f|}}
    \sum_{z\in\{0,1\}^n} f(z)\ket{z}.
    \label{eq:f_encoding}
\end{eqnarray}
Our goal is to assign edge weights to the given FBDD so that the resulting WFBDD is QSP-admissible in Def.~\ref{def:qsp_admissible_wfbdd} and its quantum-state semantics in Def.~\ref{def:q_state_wfbdd} is exactly $\ket{\phi_f}$. The cost below is the classical preprocessing cost for the supplied diagram; it does not include the cost of finding or constructing the FBDD representation of $f$.

We now define the weights using effective branch model counts. Let
\begin{eqnarray}
    X:=\{x_1,\ldots,x_n\}.
\end{eqnarray}
For every node $u$, the recursive construction associates a remaining-variable set $R(u)\subseteq X$. Intuitively, $R(u)$ is the set of variables that remain as common free variables below the node $u$ after branch-dependent skipped variables have been factored out. The terminal convention is
\begin{eqnarray}
    R(t_0)=R(t_1)=X,
    \qquad
    M(t_0)=0,
    \qquad
    M(t_1)=1.
    \label{eq:terminal_remaining_count}
\end{eqnarray}
Here $M(t_1)=1$ counts one accepting terminal pattern. Multiplicities caused by variables skipped on incoming edges are accounted for in the branch counts below.

For an internal node $u$, let $h_b(u)$ be the child reached by the $b$-edge, $b\in\{0,1\}$. After the quantities for the two children have been computed, define the common remaining-variable set
\begin{eqnarray}
    Y(u)=R(h_0(u))\cap R(h_1(u)),
    \label{eq:common_remaining_set}
\end{eqnarray}
and define
\begin{eqnarray}
    R(u)=Y(u)\setminus\{x_{i(u)}\}.
    \label{eq:remaining_variable_set}
\end{eqnarray}
Thus $Y(u)$ is the part of the remaining-variable sets that is common to the two children of $u$, and $R(u)$ is obtained after the queried variable at $u$ is removed from this common set. The set difference $R(h_b(u))\setminus Y(u)$ records the variables that are skipped along the branch from $u$ to $h_b(u)$ relative to the common variables kept for the two branches. Therefore the effective branch model count is
\begin{eqnarray}
    M_b(u)
    =
    2^{|R(h_b(u))\setminus Y(u)|}M(h_b(u)),
    \qquad b\in\{0,1\},
    \label{eq:effective_branch_count}
\end{eqnarray}
which is generally not just the child count $M(h_b(u))$. In other words, $M(u)$ is an effective model count for the subdiagram rooted at $u$ after the variables in $R(u)$ have been factored out, and $M_b(u)$ is the corresponding effective count for the $b$-branch, including the multiplicity from variables skipped along that branch. The effective model count at $u$ is then
\begin{eqnarray}
    M(u)=M_0(u)+M_1(u).
    \label{eq:node_model_count}
\end{eqnarray}

For every relevant internal node $u$ with $M(u)>0$, assign
\begin{eqnarray}
    w_b(u)=\sqrt{\frac{M_b(u)}{M(u)}},
    \qquad b\in\{0,1\}.
    \label{eq:uniform_branch_weight}
\end{eqnarray}
Then
\begin{eqnarray}
    |w_0(u)|^2+|w_1(u)|^2
    =
    \frac{M_0(u)+M_1(u)}{M(u)}
    =1.
    \label{eq:uniform_local_normalization}
\end{eqnarray}
Thus every relevant internal node has a locally normalized outgoing weight pair. 
Moreover, if the branch $u\to h_b(u)$ has no accepting continuation, then $M_b(u)=0$ and hence Eq.~(\ref{eq:uniform_branch_weight}) gives $w_b(u)=0$. 
Therefore the constructed WFBDD satisfies the local-normalization and accepting-support conditions in Def.~\ref{def:qsp_admissible_wfbdd}. Since $|f|>0$, the underlying FBDD has at least one accepting input and hence the constructed WFBDD is QSP-admissible. 
Consequently, Def.~\ref{def:q_state_wfbdd} assigns a quantum-state semantics to the constructed WFBDD. 
In the present construction, no further normalization of the edge weights is needed: the locally normalized representative is the constructed WFBDD itself. 
Hence the weights in Eq.~(\ref{eq:uniform_branch_weight}) are exactly the weights used in the amplitude products of Def.~\ref{def:q_state_wfbdd}. 
Nodes with $M(u)=0$ have no accepting continuation, so they are outside the relevant-node set and do not enter any accepting-path product.

Algorithm~\ref{alg:UWD} is an explicit bottom-up classical preprocessing procedure implementing the recursive definitions of $R(u)$, $Y(u)$, $M_b(u)$, $M(u)$, and $w_b(u)$. It is not a quantum circuit; it produces a QSP-admissible WFBDD, which is already locally normalized on all relevant internal nodes, and this WFBDD is subsequently interpreted through Def.~\ref{def:q_state_wfbdd} and prepared by Algorithm~\ref{alg:WFBDD_SG}.

\begin{figure}[!t]
\begin{algorithm}[H]
{\nolinenumbers
\caption{UniformWeightDistribution}\label{alg:UWD}
}
\begin{algorithmic}[1]
\Require A supplied $\mathrm{FBDD}_f=(V,E)$ representing $f:\{0,1\}^n\to\{0,1\}$ with $|f|>0$. For an internal node $u$, let $i(u)$ be the queried-variable index and let $h_b(u)$ be the child reached by the $b$-edge.
\Ensure A QSP-admissible $\mathrm{WFBDD}_f$ whose quantum-state semantics is $\ket{\phi_f}$ in Eq.~(\ref{eq:f_encoding}).
\State $X:=\{x_1,\ldots,x_n\}$.
\State $R(t_0):=X$, $R(t_1):=X$.
\State $M(t_0):=0$, $M(t_1):=1$.
\State Compute a reverse topological order of the internal nodes, so that both children of a node are processed before the node. \label{alg:UWD:line4}
\ForEach{internal node $u$ in this order}
    \State $Y(u):=R(h_0(u))\cap R(h_1(u))$.
    \State $R(u):=Y(u)\setminus\{x_{i(u)}\}$.
    \For{$b\in\{0,1\}$}
        \State $M_b(u):=2^{|R(h_b(u))\setminus Y(u)|}M(h_b(u))$.
    \EndFor
    \State $M(u):=M_0(u)+M_1(u)$.
    \If{$M(u)>0$}
        \For{$b\in\{0,1\}$}
            \State $w_b(u):=\sqrt{M_b(u)/M(u)}$.
        \EndFor
    \Else
        \State $w_0(u):=w_1(u):=0$. \Comment{$u$ has no accepting continuation.}
    \EndIf
\EndFor
\State Return the WFBDD obtained by assigning the weights $w_b(u)$ to the corresponding edges of $\mathrm{FBDD}_f$.
\end{algorithmic}
\end{algorithm}
\end{figure}

Since an FBDD is a directed acyclic graph with two outgoing edges per internal node, the recursive construction can be implemented with $O(|V|)$ node/edge queries to the supplied diagram. This query count should be distinguished from ordinary classical preprocessing time. If remaining-variable sets are stored explicitly, the intersections and set differences in Eqs.~(\ref{eq:common_remaining_set})--(\ref{eq:effective_branch_count}) may require $O(n|V|)$ classical time. (See Fig.~\ref{example:alg:UWD}.)

There is also an alternative normalization convention in which one uses unnormalized model-count weights, such as weights proportional to $\sqrt{M_b(u)}$, and absorbs the resulting global normalization into an additional root edge. We do not use this convention here. In the locally normalized convention used in this paper, Eq.~(\ref{eq:uniform_branch_weight}) makes every relevant outgoing pair normalized, and the free-variable prefactor in Def.~\ref{def:q_state_wfbdd} accounts for variables not queried on the selected FBDD path. Therefore no additional root edge is needed in our construction.

{\nolinenumbers
\begin{figure}[t]
\centering
\begin{minipage}[b]{0.30\linewidth}
\centering
\includegraphics[keepaspectratio, scale=0.06]{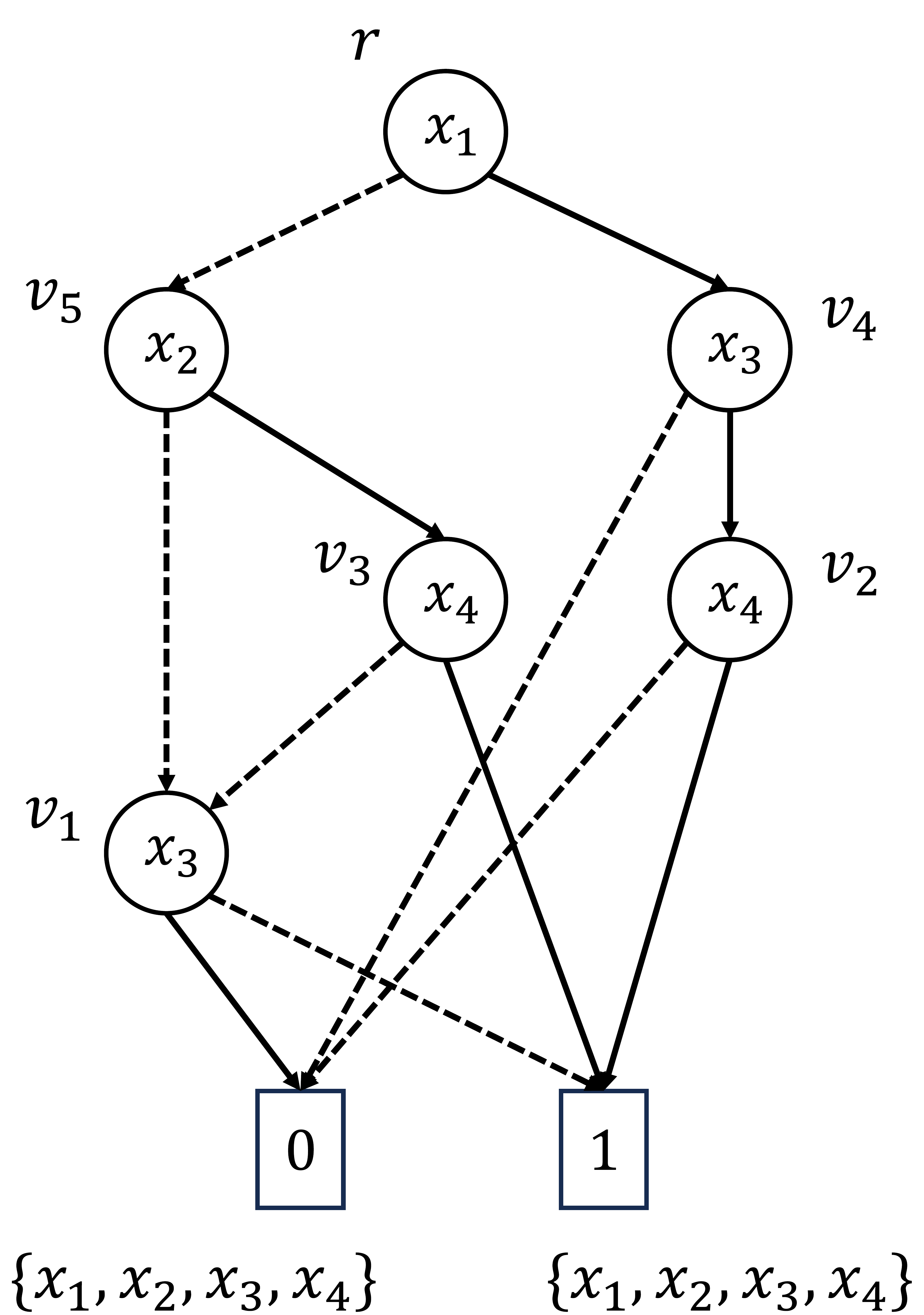}
\subcaption{}
\end{minipage}
\hfill
\begin{minipage}[b]{0.30\linewidth}
\centering
\includegraphics[keepaspectratio, scale=0.06]{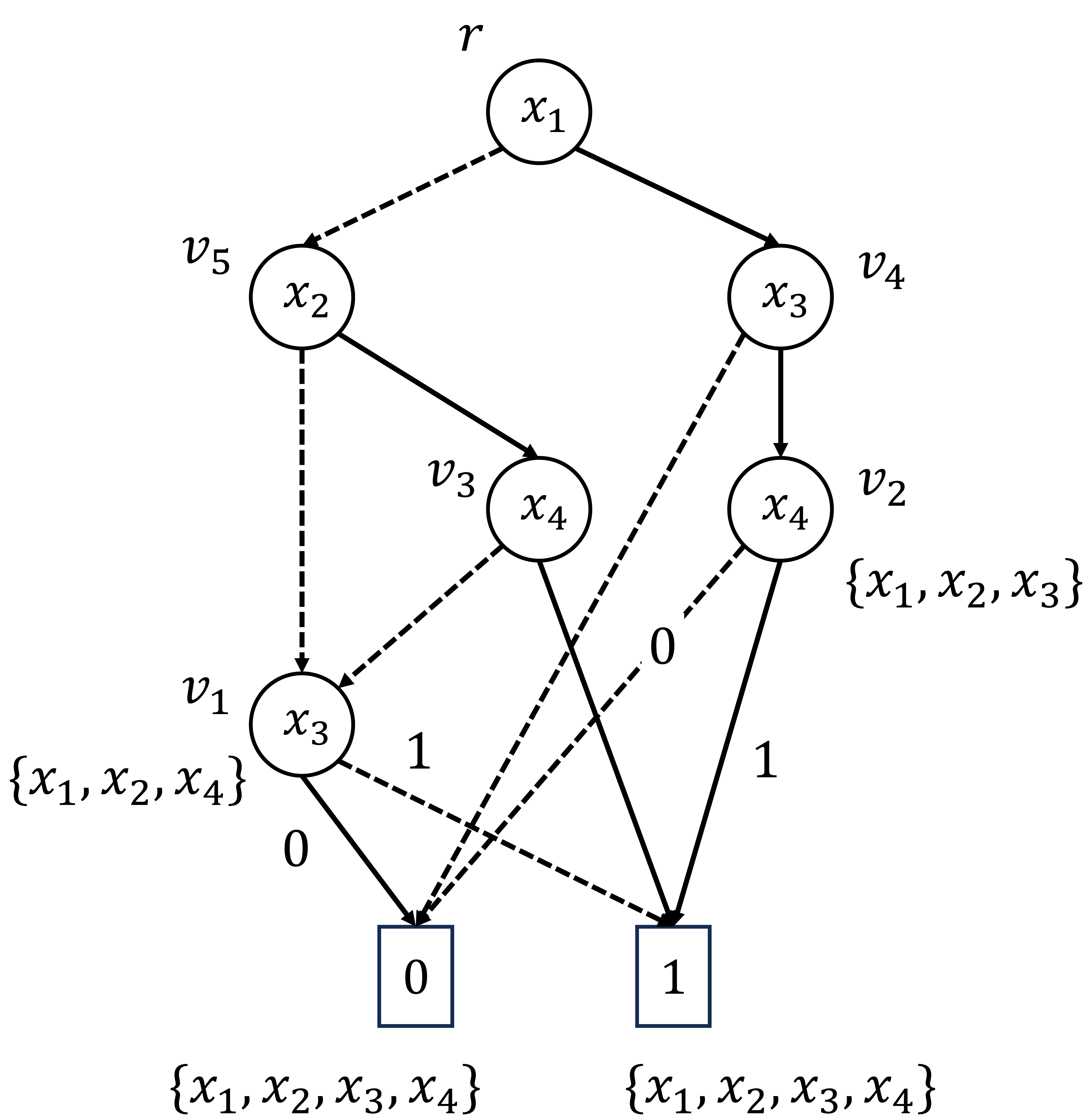}
\subcaption{}
\end{minipage}
\hfill
\begin{minipage}[b]{0.30\linewidth}
\centering
\includegraphics[keepaspectratio, scale=0.06]{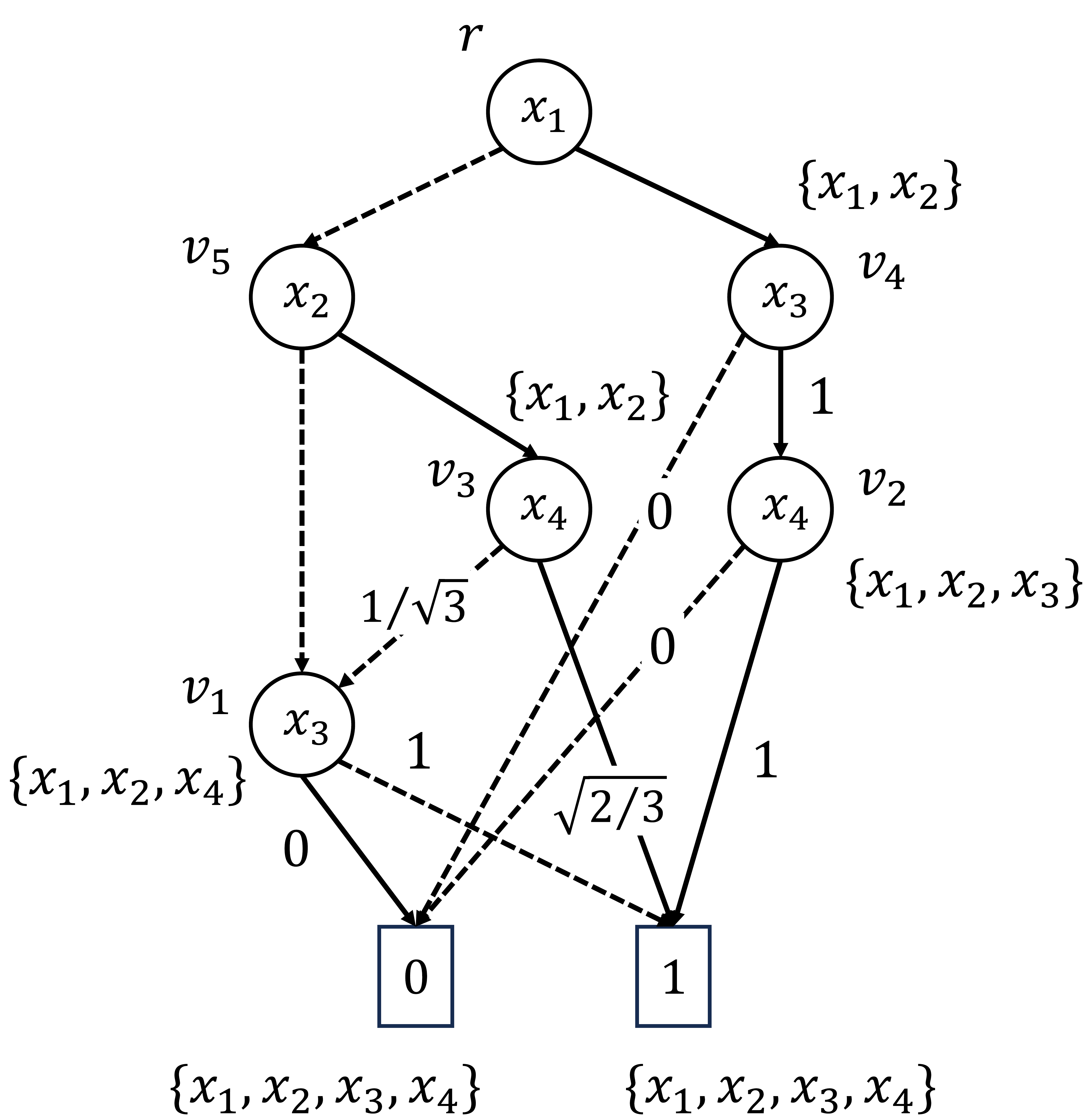}
\subcaption{}
\end{minipage}

\vspace{1mm}

\begin{minipage}[b]{0.30\linewidth}
    \centering
    \includegraphics[keepaspectratio, scale=0.06]{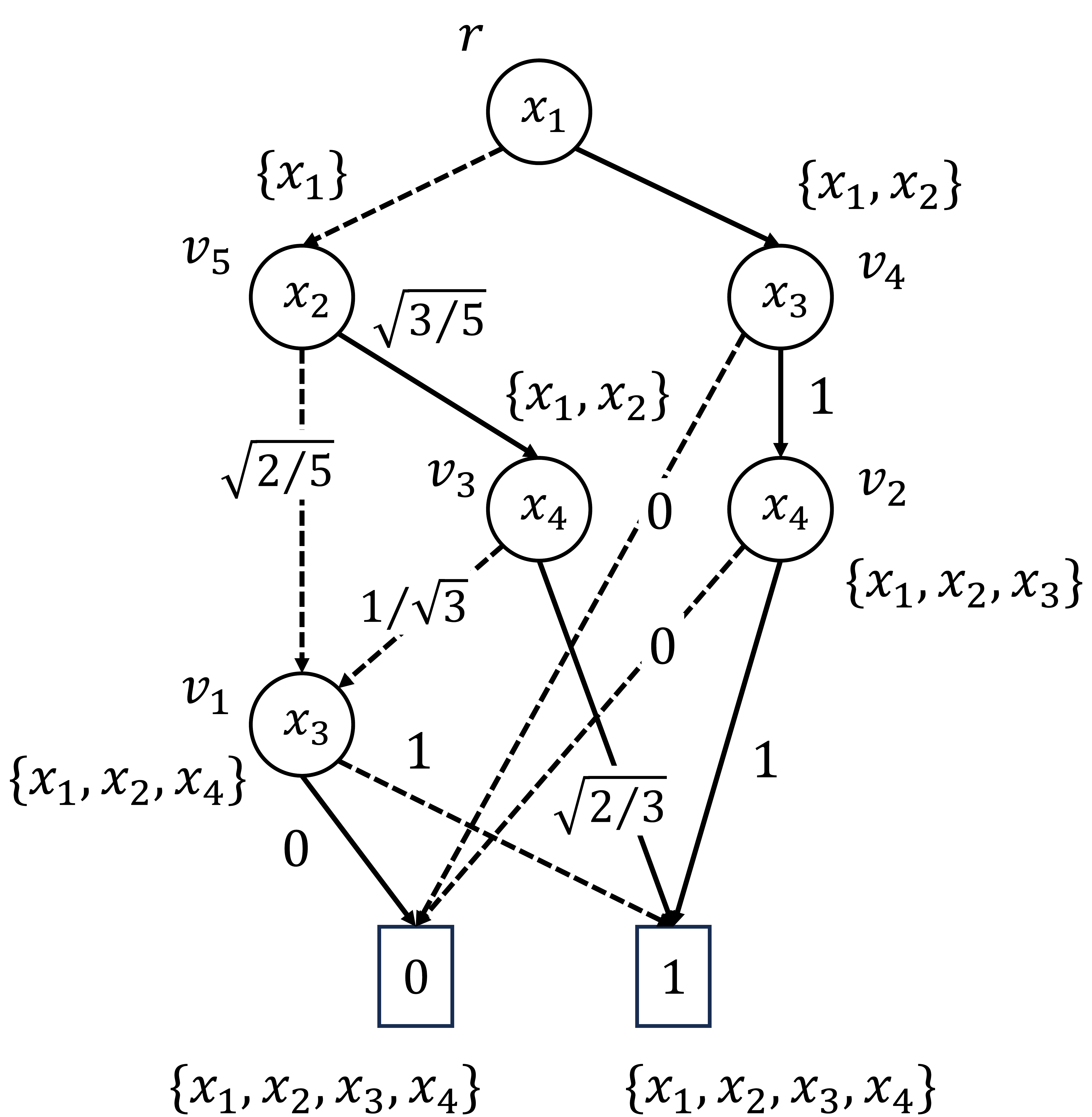}
    \subcaption{}
\end{minipage}
\hspace{0.08\linewidth}
\begin{minipage}[b]{0.30\linewidth}
    \centering
    \includegraphics[keepaspectratio, scale=0.06]{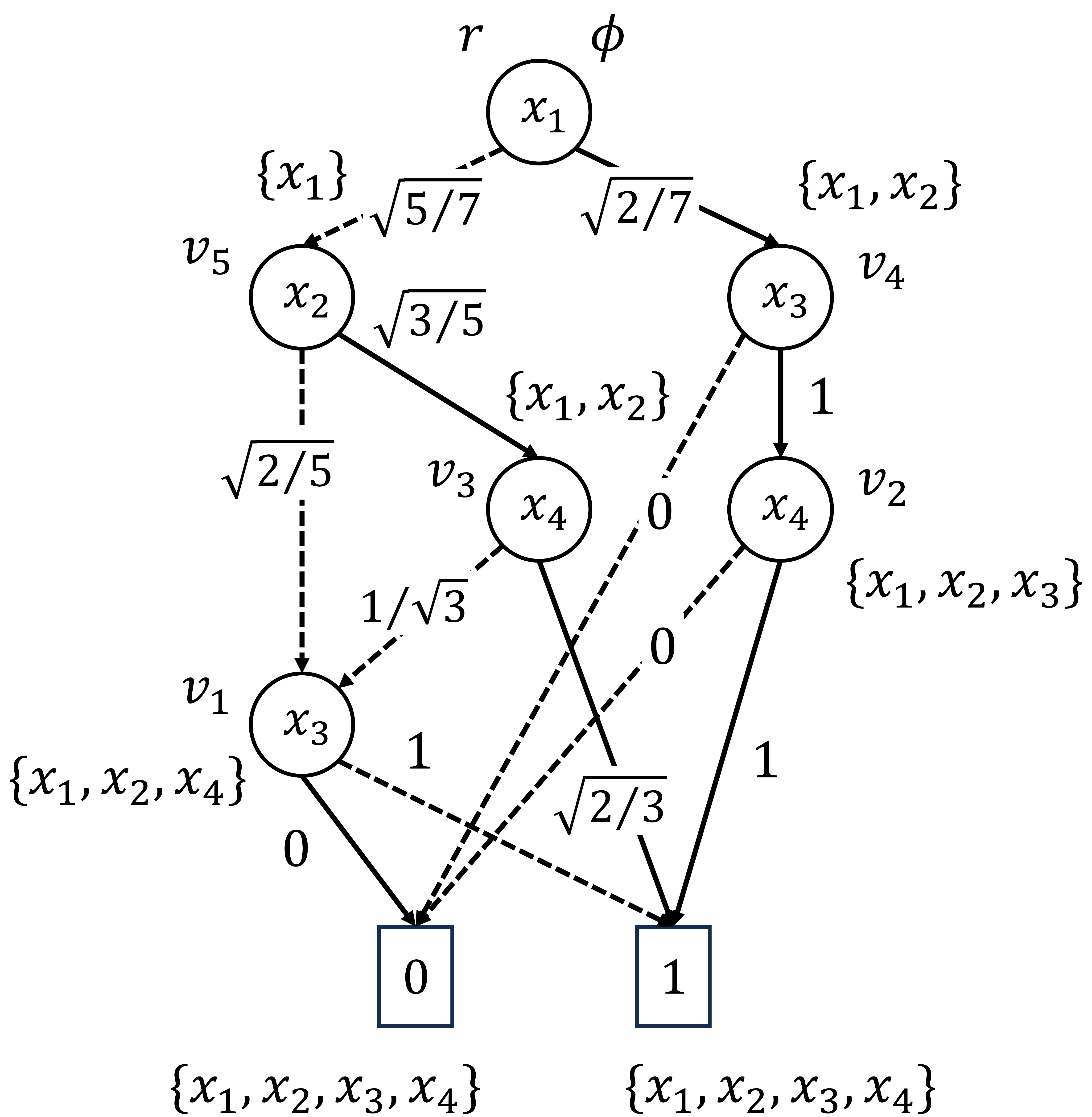}
    \subcaption{}
\end{minipage}
\caption{An example of Algorithm~\ref{alg:UWD}. 
(a) Initialization. For the terminal nodes, set $R(t_0) = R(t_1) := \{x_1,x_2,x_3,x_4\}$ and $M(t_0)=0$, $M(t_1)=1$. Here $M(t_1)=1$ counts only the accepting terminal pattern itself; multiplicities arising from variables skipped on incoming edges are accounted for in the branch counts. (b) Starting from the terminal nodes, process nodes whose children have already been assigned remaining-variable sets and model counts. In this round, the nodes $v_1$ and $v_2$ are processed. For $v_1$, the common remaining-variable set of its two children is $Y(v_1)=R(t_0)\cap R(t_1)=\{x_1,x_2,x_3,x_4\}$. Since the variable assigned to $v_1$ is $x_3$, we set $R(v_1)=Y(v_1)\setminus\{x_3\}=\{x_1,x_2,x_4\}$. The branch leading to $t_1$ has effective count $1$, whereas the branch leading to $t_0$ has effective count $0$; hence the corresponding locally normalized weights are $1$ and $0$, respectively. The node $v_2$ is processed in the same way. (c) Next, process the nodes $v_3$ and $v_4$. Consider the node $v_3$, whose $1$-edge leads to $t_1$ and whose $0$-edge leads to $v_1$ in this example. Since $Y(v_3)=R(v_1)\cap R(t_1)=\{x_1,x_2,x_4\}$, the variable $x_3$ is skipped on the branch to $t_1$. Thus the effective branch count for this branch is multiplied by $2^{|R(t_1)\setminus Y(v_3)|}=2$, giving $M_1(v_3)=2M(t_1)=2$. The other branch has effective count $M_0(v_3)=M(v_1)=1$. Therefore $M(v_3)=M_0(v_3)+M_1(v_3)=3$, and the locally normalized weights are $w_0(v_3)=1/\sqrt{3}$ and $w_1(v_3)=\sqrt{2/3}$. The node $v_4$ is processed analogously. (d), (e) Repeating this bottom-up procedure assigns the remaining-variable set $R(u)$, the effective branch counts $M_b(u)$, the total count $M(u)=M_0(u)+M_1(u)$, and the locally normalized weights $w_b(u)=\sqrt{M_b(u)/M(u)}$ to all relevant internal nodes. Branches with no accepting continuation have effective count zero and hence receive weight zero.}
\label{example:alg:UWD}
\end{figure}
}

\begin{theorem}[Uniform-superposition WFBDD from an FBDD]
\label{lem:state_wfbdd}
\label{thm:uniform_wfbdd}
Let $\mathrm{FBDD}_f=(V,E)$ be a supplied FBDD representing a Boolean function $f:\{0,1\}^n\to\{0,1\}$ with $|f|>0$. Algorithm~\ref{alg:UWD} constructs a QSP-admissible $\mathrm{WFBDD}_f$ by assigning weights to the relevant internal nodes of the supplied FBDD, optionally identifying zero-count subdiagrams with $t_0$. The quantum state associated with the resulting WFBDD by Def.~\ref{def:q_state_wfbdd} is the uniform superposition $\ket{\phi_f}$ in Eq.~(\ref{eq:f_encoding}). The construction uses $O(|V|)$ node/edge queries to the supplied diagram. With explicit set representations, the classical set-processing time is $O(n|V|)$.
\end{theorem}

\begin{proof}[Proof sketch]
The effective branch count $M_b(u)$ counts accepting continuations below the child $h_b(u)$ together with the multiplicity of variables skipped by the branch $u\to h_b(u)$. Therefore the weights in Eq.~(\ref{eq:uniform_branch_weight}) locally normalize every relevant outgoing pair, and any branch with no accepting continuation receives weight zero; hence the constructed WFBDD is QSP-admissible. Along any accepting root-to-$t_1$ path, the product of the ratios $M_b(u)/M(u)$ telescopes, namely, the intermediate factors cancel successively so that only the endpoint counts remain, while the free-variable prefactor in Def.~\ref{def:q_state_wfbdd} cancels the skipped-variable factors. Consequently every satisfying assignment receives the same amplitude $1/\sqrt{|f|}$, and every rejecting assignment receives amplitude zero. The construction uses one reverse-topological pass through the supplied diagram, giving the stated query complexity and explicit-set preprocessing time. The detailed counting proof is given in Appendix~\ref{apdx:proof_uniform}.
\end{proof}

\section{Controlled Phase operation based on FBDD} \label{Sec5}

In this section, we show another encoding of a Boolean function into a quantum state.
For a given $\mathrm{FBDD}_f$, we show that Algorithm~\ref{alg:WFBDD_PG} constructs a quantum circuit realizing a unitary such that $U_f\ket{x} = e^{i\theta f(x)}\ket{x}$.
Using such a quantum circuit, we can trivially prepare an $n$-qubit quantum state 
\begin{eqnarray}
    \ket{\psi_{\theta}} := \frac{1}{\sqrt{2^n}} \sum_{x \in \{0,1\}^n} e^{i\theta f(x)} \ket{x},
\end{eqnarray}
from a uniform superposition $(H\ket{0})^{\otimes n}$.

Algorithm~\ref{alg:WFBDD_PG} is derived by modifying Algorithm~\ref{alg:WFBDD_SG} slightly.
In Algorithm~\ref{alg:WFBDD_SG}, notice that the edge-label sequence $z = z_1\cdots z_n$ of a path from the root $r$ to the terminal node $t_1$ represents the basis state $\ket{z}$ in the support $\{ \ket{x} \ | \ x \in \{ 0,1\}^n, f(x) = 1 \}$, and the nodes on the path are recorded in the ancillary qubits. 
Then, to construct the quantum circuit realizing a unitary such that $U_f\ket{x} = e^{i\theta f(x)}\ket{x}$, first, let the $n$ qubits assigned to the variables $x_1,\ldots, x_n$ be the input quantum registers (or replace Line~\ref{alg:SG3} of Algorithm~\ref{alg:WFBDD_SG} with Line~\ref{alg:SG3} of Algorithm~\ref{alg:WFBDD_PG}) and assign an ancillary qubit to the terminal node $t_1$ on which the result $f(x)$ is output (Line~\ref{alg:PG6} of Algorithm~\ref{alg:WFBDD_PG}). 
Second, remove the controlled-$H$ gates and controlled-$U(u)$ gates (or remove Lines~\ref{alg:SG10} and \ref{alg:SG11} of Algorithm~\ref{alg:WFBDD_SG}) and remove the terminal node $t_1$ from the condition in Line~\ref{alg:SG17} (or see Line~\ref{alg:PG16} of Algorithm~\ref{alg:WFBDD_PG}), thereby adding Toffoli gates that output $f(x)$ to the last ancillary qubit $t_1$. 
Finally, apply a phase gate to the ancillary qubit $t_1$ (Line~\ref{alg:PG23}) and undo all the Toffoli gates. (See also Fig.~\ref{example:alg:PG} as an example of Algorithm~\ref{alg:WFBDD_PG}.)
We give a corollary as a summary of the above discussion.
\begin{corollary}
    For a given $\mathrm{FBDD}_f = (V,E)$ encoding an $n$-bit Boolean function $f$, Algorithm~\ref{alg:WFBDD_PG} gives a quantum circuit $U_f$ consisting of a phase gate and at most $4( |V|-1 )$ Toffoli gates with $|V|-1$ ancillary qubits such that for $x \in \{ 0,1 \}^n$,
    \begin{eqnarray}
        U_f : \ket{x} \otimes \ket{1} \otimes \ket{0}^{\otimes |V|-2} \mapsto e^{i\theta f(x)} \ket{x} \otimes \ket{1} \otimes \ket{0}^{\otimes |V|-2}.
    \end{eqnarray}
\end{corollary}

\begin{figure}[!t]
\begin{algorithm}[H]
{\nolinenumbers
\caption{ControlledPhase} \label{alg:WFBDD_PG}
}
\begin{algorithmic}[1]
\Require A $\mathrm{FBDD}_f = (V,E)$ in Def.~\ref{def:fbdd}, where let $i(u)$, $e_b(u)$, and $h_b(u)$ denote the index of the variable assigned to an internal node $u$, the $b$-edge out from $u$, and the child node of $e_b(u)$, respectively.
\Ensure A quantum circuit realizing a unitary such that $U_f \ket{x} = e^{i\theta f(x)}\ket{x}$. 
\State {\bf [Initialization]} \label{alg:PG1}
\State Prepare a set $Q = \{ r \}$ and an empty stack $S$. \label{alg:PG2}
\State Assign qubits for an input state to the variables $x_1,\ldots, x_n$. \label{alg:PG3}
\State Assign a qubit initialized to $\ket{1}$ to the root $r$. \label{alg:PG4}
\State Assign qubits initialized to $\ket{0}$ to the internal nodes except the root $r$. \label{alg:PG5}
\State Assign a qubit initialized to $\ket{0}$ to the terminal node $t_1$. \label{alg:PG6}
\State {\bf [Main loop]} \label{alg:PG7}
\While{the set $Q$ is not empty} \label{alg:PG8}
    \State Remove a node $u$ from $Q$. \label{alg:PG9}
    \State Push the node $u$ to the stack $S$. \label{alg:PG10}
    \ForEach{$b$-edge $e_b(u)$ out from the node $u$ and its child node $h_b(u)$} \label{alg:PG11}
        \State Remove the $b$-edge $e_b(u)$ from $E$. \label{alg:PG12}
        \If{the child node $h_b(u)$ has no incoming edge and is not a terminal node} \label{alg:PG13}
            \State Add the child node $h_b(u)$ to the set $Q$. \label{alg:PG14}
        \EndIf \label{alg:PG15}
        \If{$h_b(u) \neq t_0$} \label{alg:PG16}
            \State Apply the $X$ gate to the qubit $x_{i(u)}$ when $b = 0$. \label{alg:PG17}
            \State Apply the Toffoli gate to the qubit $h_b(u)$ controlled by the qubits $u$ and $x_{i(u)}$. \label{alg:PG18}
            \State Apply the $X$ gate to the register assigned to $x_{i(u)}$ when $b = 0$. \label{alg:PG19}
        \EndIf \label{alg:PG20}
    \EndFor \label{alg:PG21}
\EndWhile \label{alg:PG22}
\State Apply a phase gate to the qubit $t_1$. \label{alg:PG23}
\State {\bf [Undo]} \label{alg:PG24}
\While{the stack $S$ is not empty} \label{alg:PG25}
    \State Pop a node $u$ from the stack $S$. \label{alg:PG26}
    \State For each original $b$-edge $e_b(u)$ out from $u$, do Lines~\ref{alg:PG16}--\ref{alg:PG20}.  \label{alg:PG27}
\EndWhile \label{alg:PG28}
\end{algorithmic}
\end{algorithm}
\end{figure}

{\nolinenumbers
\begin{figure}[t]
    \begin{minipage}[b]{0.4\linewidth}
        \centering
        \includegraphics[keepaspectratio, scale=0.15]{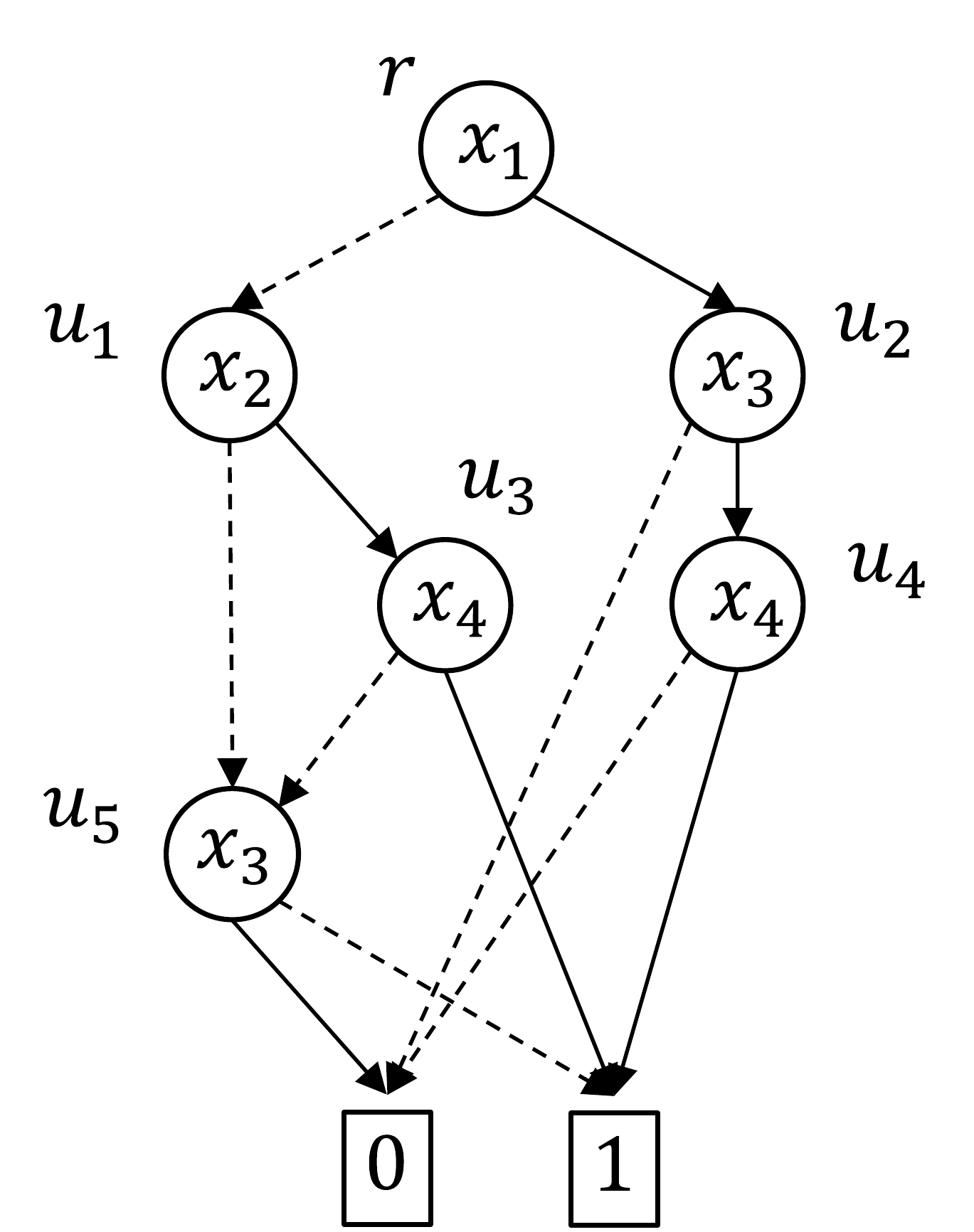}
        \subcaption{}
    \end{minipage}
    \begin{minipage}[b]{0.5\linewidth}
        \centering
        \includegraphics[keepaspectratio, scale=0.15]{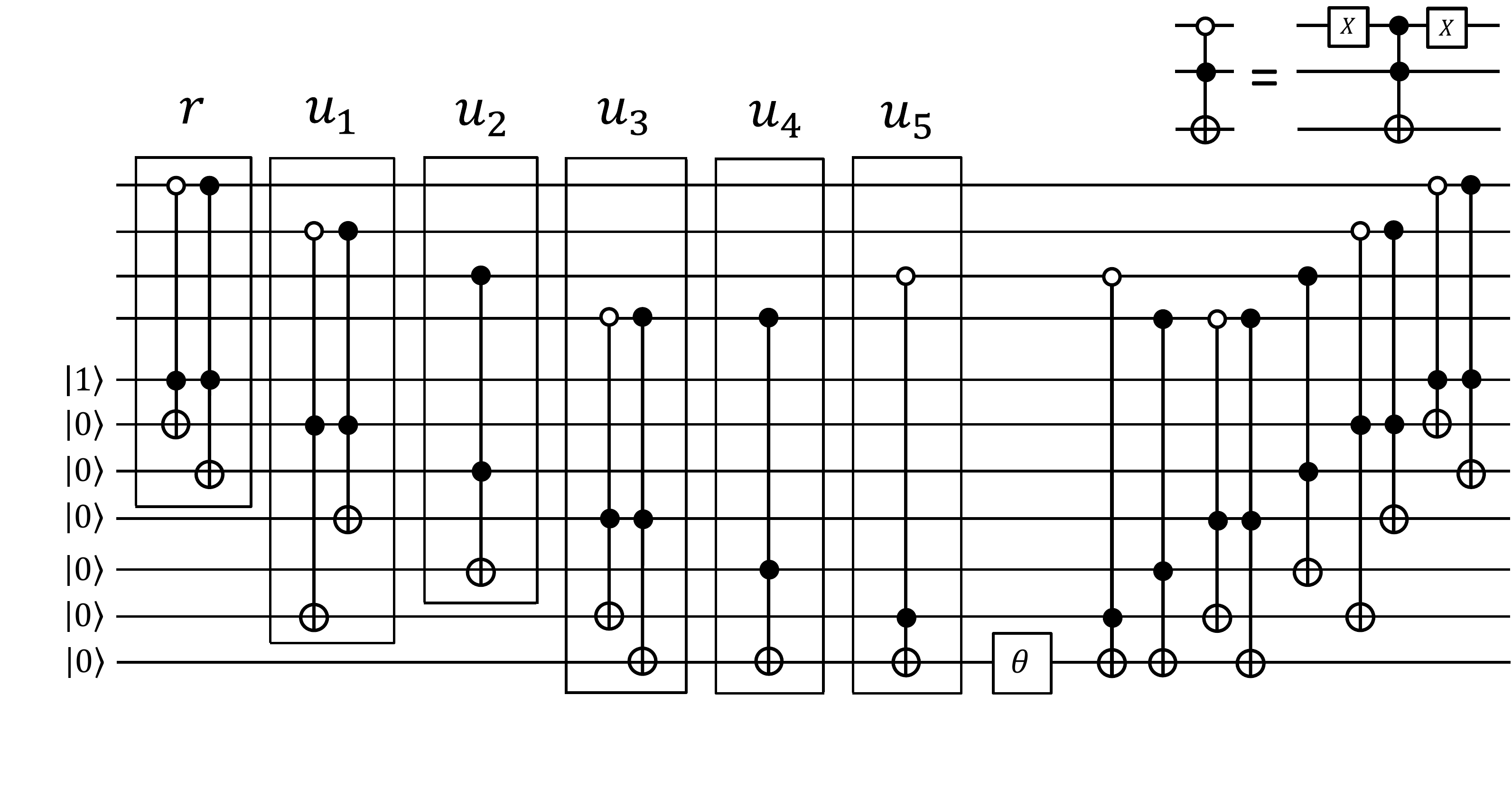}
        \subcaption{}
    \end{minipage}
    \caption{An example of Algorithm \ref{alg:WFBDD_PG}. (a) An example of $\textrm{FBDD}_f$. In Algorithm \ref{alg:WFBDD_PG}, all the internal nodes are pushed to a stack $S$ in a topological ordering such that for every directed edge $(u, v)$ out from $u$ to $v$, $u$ comes before $v$, {\it i.e.}, $r \prec u_1 \prec \cdots \prec u_5$. Start from the root $r$ and push $r$ to $S$. Then, remove the $0$-edge $(r, u_1)$ and the $1$-edge $(r, u_2)$, at which, if the indegree of $u_{1(2)}$ is zero, $u_{1(2)}$ is added to a set $Q$. The added nodes are pushed to $S$ in the following rounds and are ordered. Repeating this process, we order all the internal nodes and push them to $S$. (b) The quantum circuit generated from the $\textrm{FBDD}_f$. Initially, assign qubits to the variables $x_1,\ldots, x_4$, the internal nodes $r, u_1,\ldots, u_5$, and the terminal node $t_1$. Each qubit $x_i$ serves as an input qubit, the qubit $r$ is set to $\ket{1}$, and each qubit $u_j$ is set to $\ket{0}$. Then, the Toffoli gate controlled by the qubits $r$ and $x_1$ is applied to each qubit $u_{1(2)}$ as shown in the figure. Applying the Toffoli gates for rest of the internal nodes in the topological ordering, we derive the quantum circuit in the figure.} \label{example:alg:PG}
\end{figure}
}

\section{Application to Black-Box QSP} \label{Sec6}

In this section, we give an $n$-qubit quantum state for which the cost expression of the black-box QSP algorithm of Ref.~\cite{Bausch2022fastblackboxquantum} is exponential, while the state can be generated efficiently by our proposed WFBDD-based state-preparation procedure. 
For a fixed constant $\delta \in (0,1)$, we consider an $n$-qubit quantum state
\begin{eqnarray}
    \ket{\psi_{\delta}} = \frac{1}{|| \alpha ||_2} \sum_{x \in \{ 0,1 \}^n} \alpha_x \ket{x}, \label{eq:bb_state}
\end{eqnarray}
where $\alpha_x = \delta^{w_H(x)} / \binom{n}{w_H(x)}$, and $w_H(x)$ is the Hamming weight of $x$. We show that the cost expression for the number of amplitude-amplification steps required to generate the quantum state $\ket{\psi_{\delta}}$ is exponential for the black-box QSP algorithm of Ref.~\cite{Bausch2022fastblackboxquantum}, and we give a weighted FBDD describing $\ket{\psi_{\delta}}$. 

The black-box QSP in Ref.~\cite{Bausch2022fastblackboxquantum} transforms an initial state through an intermediate state to the target state by performing the nested amplitude amplifications. Furthermore, for an initial state and a target state 
\begin{eqnarray}
    \ket{s} &:=& \frac{1}{\lVert \bar{A} \rVert_2} \sum_{j=0}^{g-1} \bar{A}_j \ket{j}, \\
    \ket{t} &:=& \frac{1}{\lVert \alpha \rVert_2} \sum_{i=0}^{N-1} \alpha_i \ket{i},
\end{eqnarray}
where
\begin{eqnarray}
    \bar{A}_j &:=& \sum_{i=0}^{N-1} b_{ij} \sqrt{w_j}, \\
    \alpha_i &:=& \sum_{j=0}^{g-1} b_{ij} w_j,
\end{eqnarray}
and $B=(b_{ij})$ is some Boolean matrix, the number of nested amplitude amplifications is shown to be evaluated as $O\left(\sqrt{N}\lVert \alpha \rVert_1 / \lVert \bar{A} \rVert_2\right)$~\cite{Bausch2022fastblackboxquantum}.
Set $N = 2^n$, $g = n+1$, and 
\begin{eqnarray}
    b_{ij} &=& 
    \left\{
        \begin{array}{ll}
            1 & w_H(i) = j, \\
            0 & \text{otherwise}.
        \end{array}
    \right. \\
    w_j &=& \frac{\delta^j}{\binom{n}{j}},
\end{eqnarray}
we derive that
\begin{eqnarray}
    \bar{A}_j &=& \left(\binom{n}{j} \delta^j \right)^{\frac{1}{2}}, \\
    \alpha_i &=& \frac{\delta^{w_H(i)}}{\binom{n}{w_H(i)}},
\end{eqnarray}
and then
\begin{eqnarray}
    \sqrt{N} \frac{\lVert \alpha \rVert_1}{\lVert \bar{A} \rVert_2} = \frac{1-\delta^{n+1}}{1-\delta} \left( \frac{2}{1+\delta} \right)^{\frac{n}{2}}.
\end{eqnarray}
Thus, for any fixed constant $\delta\in(0,1)$, this cost expression for the nested amplitude amplifications is exponential in $n$.

{\nolinenumbers
\begin{figure}[t]
    \centering
    \includegraphics[keepaspectratio,scale=0.2]{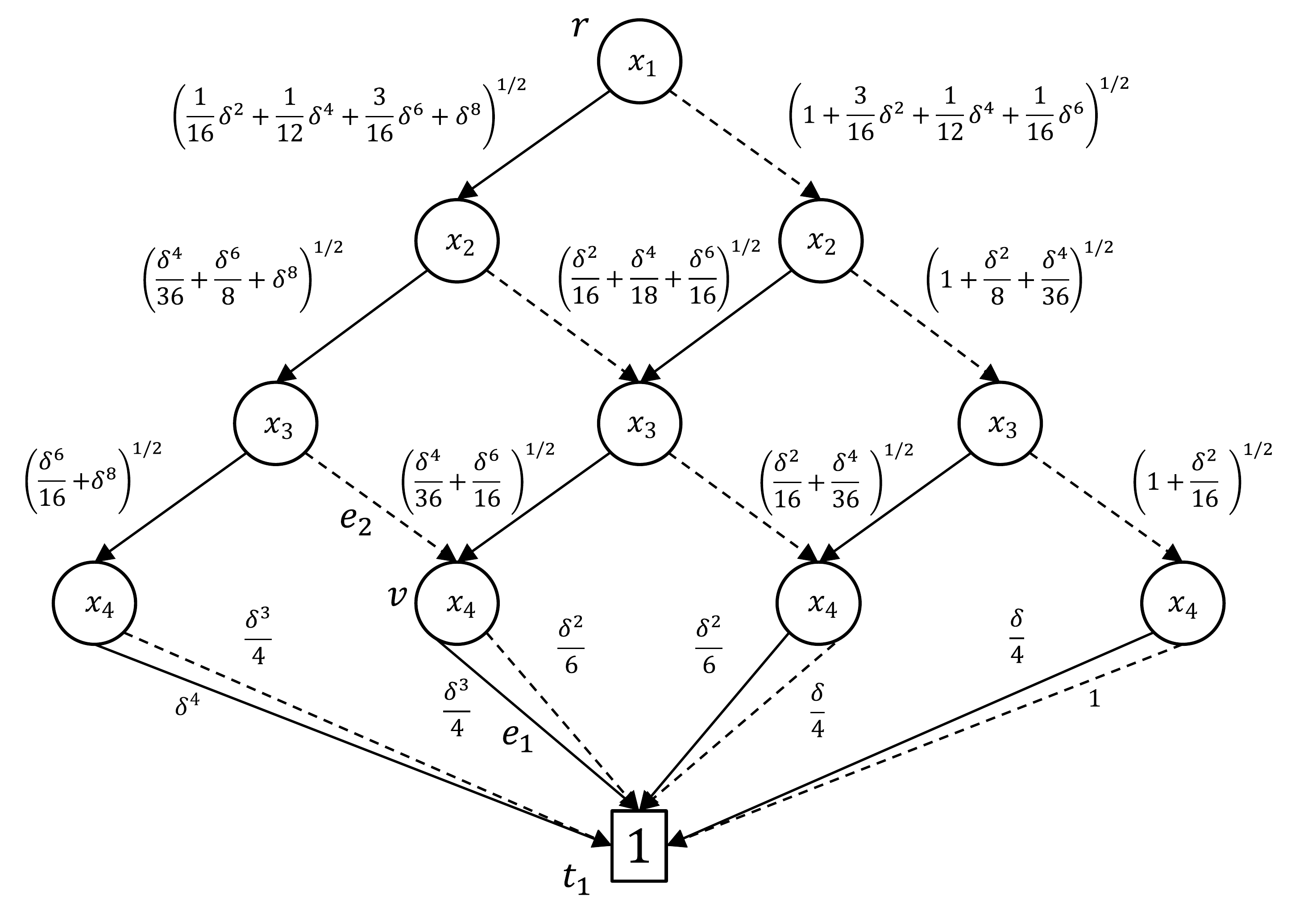}
    \caption{A weighted OBDD describing a $4$-qubit quantum state in Eq.~(\ref{eq:bb_state}). The weighted OBDD consists of $4$ layers, where the $i$-th layer has $i$ nodes labeled with a variable $x_i$. For $i = 1,2,3$ and $j \in [i]$, the $j$-th node in the $i$-th layer connects to the $j$-th node in the $(i+1)$-th layer with the $1$-edge and connects to the $(j+1)$-th node in the $(i+1)$-th layer with the $0$-edge. Any node in the $4$-th layer connects to the $1$-terminal node with both of the $0$-edge and the $1$-edge. Weights are assigned as follows. For example, for an edge $e_1$ going into the terminal node $t_1$, notice that any path from the root $r$ to $t_1$ through the edge $e_1$ passes through three $1$-edges. Thus, remembering that the coefficient $\alpha_x$ is determined by the Hamming weight $w_H(x)$, we can assign $\alpha_x = \delta^3/4$ on the edge $e_1$. The weight of another unweighted edge, {\it e.g.}, $e_2$, is $\left( \delta^4/36 + \delta^6/16 \right)^{1/2}$, the square root of the sum of the squared weights of the edges outgoing from the node $v$ into which $e_2$ goes. The displayed edge weights are raw subtree-norm weights; the QSP-admissible WFBDD used by the state-preparation procedure is obtained by applying the local normalization in Eq.~(\ref{eq:local_normalized_weight}).} \label{fig:4comb_wfbdd}
\end{figure}
}

On the other hand, we show that the $n$-qubit quantum state $\ket{\psi_{\delta}}$ can be written in the form of a weighted OBDD with $n(n+1)/2$ internal nodes. 
Then, from Theorem~\ref{thm:WS}, Algorithm~\ref{alg:WFBDD_SG} gives the quantum state $\ket{\psi_{\delta}}$ with $O(n(n+1))$ uses of single- or two-qubit gates and $n(n+1)/2$ ancillary qubits.

The weighted OBDD consists of $n$ layers, where the $i$-th layer has $i$ nodes labeled with a variable $x_i$. (See Fig.~\ref{fig:4comb_wfbdd} in the case of $4$ qubits.) 
For $i \in [n-1]$ and $j \in [i]$, the $j$-th node in the $i$-th layer connects to the $j$-th node in the $(i+1)$-th layer with the $1$-edge and connects to the $(j+1)$-th node in the $(i+1)$-th layer with the $0$-edge. Any node in the $n$-th layer connects to the $1$-terminal node with both of the $0$-edge and the $1$-edge. We include an unreachable $0$-terminal $t_0$ only to match the terminal convention used in the definitions and algorithms.

Weights are assigned as follows. 
First, for an edge $e$ going into the terminal node $t_1$, notice that any path from the root $r$ to $t_1$ through the edge $e$ passes through the same number of $1$-edges. Thus, remembering that the coefficient $\alpha_x$ is determined by the Hamming weight $w_H(x)$, we can assign each $\alpha_x$ on each edge going into the terminal node $t_1$ based on the Hamming weight $w_H(x)$ of $x$ and the number of $1$-edges passed. The weight of the other unweighted edge is the square root of the sum of the squared weights of the edges outgoing from the node into which the unweighted edge goes. 

The weights assigned in this way are raw subtree-norm weights and are not
necessarily locally normalized. We therefore apply the local normalization
of Eq.~(\ref{eq:local_normalized_weight}) to obtain the locally normalized
representative $\widetilde G$. More explicitly, if $S(u)$ denotes the sum
of the squared moduli of the amplitudes represented by the continuations
below $u$, then the raw weight of the edge $u\to h_b(u)$ is
$\sqrt{S(h_b(u))}$, with the terminal-layer convention that the edge
entering $t_1$ carries the corresponding coefficient $\alpha_x$. Hence,
after local normalization,
\begin{eqnarray}
    \widetilde w_b(u)
    =
    \sqrt{\frac{S(h_b(u))}{S(u)}},
\end{eqnarray}
with the analogous terminal-edge expression. Along the path corresponding
to $x$, these normalization factors telescope and yield
\begin{eqnarray}
    \prod_{u\in V_x\setminus\{t_0,t_1\}}
    \widetilde w_{x_{i(u)}}(u)
    =
    \frac{\alpha_x}{\sqrt{S(r)}}
    =
    \frac{\alpha_x}{\|\alpha\|_2},
\end{eqnarray}
Since every root-to-terminal path in this example reaches $t_1$, the
locally normalized representative $\widetilde G$ is QSP-admissible and its
quantum-state semantics is exactly $|\psi_\delta\rangle$.

\section{Application to the block encoding} \label{Sec7}

Our techniques for QSP can be applied to the block encoding, a technique of encoding a matrix as a block of a unitary, defined as follows (Definition 47 of Ref.~\cite{Gilyen2018-gs}).
\begin{definition}
    Let $A$ be an $s$-qubit operator, $\alpha, \epsilon \in \mathbb{R}_+$, and $a \in \mathbb{N}$. We say that the $(s+a)$-qubit unitary $U$ is an $(\alpha, a, \epsilon)$-block encoding of $A$ if
    \begin{eqnarray}
        || A - \alpha(\bra{0}^{\otimes a} \otimes I) U (\ket{0}^{\otimes a} \otimes I) || \le \epsilon,
    \end{eqnarray}
    where $|| \cdot ||$ is the operator norm.
    Note that since $|| U || = 1$, we necessarily have $|| A || \le \alpha + \epsilon$.
\end{definition}
References~\cite{Gilyen2018-gs, 10.1145/3313276.3316366} give constructions of block encodings of density operators, POVM operators, Gram matrices, and sparse-access matrices, respectively (Lemma 45 - 48 of Ref.~\cite{Gilyen2018-gs}). 
Therefore, when the unitary operators or the oracles assumed in each lemma of Ref.~\cite{Gilyen2018-gs} are given through weighted FBDDs, we can construct the block encodings.

For example, we consider a $\mathrm{WFBDD}_f$ representing an $n$-qubit quantum state $\ket{\psi_f} = \sum_z c_z f(z) \ket{z}$. 
From Algorithm~\ref{alg:WFBDD_SG}, we can implement a quantum circuit $U_f$ such that 
\begin{eqnarray}
    U_f : \ket{0}^{\otimes n} \otimes \ket{1}_r \otimes \ket{0}^{\otimes |V|-3} \mapsto \sum_{z \in \{ 0,1 \}^n} c_z f(z) \ket{z} \otimes \ket{1}_r \otimes \ket{0}^{\otimes |V|-3},
\end{eqnarray}
where $|V|$ is the number of nodes.
Then, $(U_f^{\dagger} \otimes I_n)\textsc{SWAP}(U_f \otimes I_n)$ is the $(1,n+|V|-2,0)$-block encoding of $\ket{\psi_f}\bra{\psi_f}$ with respect to the fixed ancillary state $\ket{0}^{\otimes n}\ket{1}_r\ket{0}^{\otimes |V|-3}$, where $\textsc{SWAP}$ is the operation swapping the first $n$ qubits and the last $n$ qubits. Equivalently, conjugating by a Pauli-$X$ on the root ancillary qubit converts this to the standard all-zero ancillary convention.
As shown in Lemma 45 of Ref.~\cite{10.1145/3313276.3316366}, we have that
\begin{eqnarray}
    && \bra{0}^{\otimes n}\bra{1}_r\bra{0}^{\otimes |V|-3} \otimes \bra{x} (U_f^{\dagger} \otimes I_n)\textsc{SWAP}(U_f \otimes I_n) \ket{0}^{\otimes n}\ket{1}_r\ket{0}^{\otimes |V|-3} \otimes \ket{y} \nonumber \\
    &=& \left( \sum_z c^{\ast}_z f(z) \bra{z} \otimes \bra{1}_r \otimes \bra{0}^{\otimes |V|-3} \otimes \bra{x} \right) \textsc{SWAP} \left( \sum_{z'} c_{z'} f(z') \ket{z'} \otimes \ket{1}_r \otimes \ket{0}^{\otimes |V|-3} \otimes \ket{y} \right) \\
    &=& \left( \sum_z c^{\ast}_z f(z) \bra{z} \otimes \bra{1}_r \otimes \bra{0}^{\otimes |V|-3} \otimes \bra{x} \right) \left( \sum_{z'} c_{z'} f(z') \ket{y} \otimes \ket{1}_r \otimes \ket{0}^{\otimes |V|-3} \otimes \ket{z'} \right) \\
    &=& \sum_{z,z'} c^{\ast}_z c_{z'} f(z)f(z') \braket{z|y}\braket{x|z'} \\
    &=& \braket{x | \psi_f} \braket{\psi_f | y}.
\end{eqnarray}

As another example, for the above $\mathrm{WFBDD}_f$, if we can efficiently implement controlled versions of the unitaries $V_z \in SU(2^n)$, we can construct a $(1, n +|V|-2, 0)$-block encoding of a $2^n \times 2^n$ matrix 
\begin{eqnarray}
    A = \sum_{ z \in \{ 0,1 \}^n } |c_z|^2 f(z) V_z,
\end{eqnarray}
by following the technique of the block encoding of Gram matrices (Lemma 48 of \cite{Gilyen2018-gs}).
Specifically, we can construct two unitaries $U_L$ and $U_R$ such that
\begin{eqnarray}
    U_L : \ket{0^n} \otimes \ket{1}_r \otimes \ket{0}^{\otimes |V|-3} \otimes \ket{y} &\mapsto& \sum_{z \in \{ 0,1 \}^n} c_z f(z) \ket{z} \otimes \ket{1}_r \otimes \ket{0}^{\otimes |V|-3} \otimes V(z) \ket{y}, \\
    U_R : \ket{0^n} \otimes \ket{1}_r \otimes \ket{0}^{\otimes |V|-3} \otimes \ket{x} &\mapsto& \sum_{z \in \{ 0,1 \}^n} c_z f(z) \ket{z} \otimes \ket{1}_r \otimes \ket{0}^{\otimes |V|-3} \otimes \ket{x},
\end{eqnarray}
and then $U_R^{\dagger}U_L$ is the block encoding of $A$ in the same fixed-ancilla convention.

The above two results can give non-trivial examples of block encodings of non-sparse matrices.
For example, the block encodings of matrices derived from the $\mathrm{WFBDD}_h$, where $h$ is the Boolean function in Eq.~(\ref{eq:dense_h}), provide such examples, since the $\mathrm{WFBDD}_h$ represents an $m$-qubit quantum state with $d=2^{\Omega(m)}$ nonzero coefficients, where $m=3n+2$ in the notation of Section~\ref{Sec3}, as shown in Section~\ref{Sec3}.

\section{Conclusion} \label{Sec8}

In this paper, we proposed a quantum state-preparation (QSP) circuit construction for the setting in which the classical description of an $n$-qubit quantum state can be represented by a weighted free binary decision diagram ($\mathrm{WFBDD}_f$)~\cite{Breitbar_P1995-dg} whose locally normalized representative is QSP-admissible. We analyzed the computational complexity of the proposed construction and constructed a quantum circuit preparing the $n$-qubit quantum state with $O(N)$ uses of single- or two-qubit gates and $N-2$ ancillary qubits.

The efficiency of our state-preparation procedure depends on the compression rate of WFBDDs. 
More precisely, the efficiency statements assume that a compact WFBDD representation whose locally normalized representative is QSP-admissible is already available. The cost of discovering or constructing such a representation is not included in the circuit complexity. When coherent access to diagram data is required, a qROM/QRAM-style data-access mechanism can be used, and its construction cost should be accounted for separately.
When $N = O(\mathrm{poly}(n))$, our construction gives an $O(\mathrm{poly}(n))$-sized QSP circuit with $\mathrm{poly}(n)$ ancillary qubits, which reduces the ancillary qubits needed in existing QSP methods for a general class of quantum states~\cite{Zhang2022-bb,10.1109/TCAD.2023.3244885,Rosenthal2021-pj}.
Furthermore, we gave an example of an $O(\mathrm{poly}(n))$-sized $\mathrm{WFBDD}_f$ whose locally normalized representative is QSP-admissible and describes an $m$-qubit quantum state with $2^{\Omega(m)}$ nonzero coefficients for $m=\Theta(n)$.
From these results, we identify a new subclass of efficiently preparable quantum states based on WFBDDs.

We gave three examples of QSP tasks that our FBDD-based state-preparation procedure can resolve efficiently.
First, we studied preparing a uniform superposition over a Boolean function, {\it i.e.}, $\sum_{x\in \{0,1\}^n} f(x) \ket{x} / \sqrt{|f|}$.
We showed that there exists an $O(N)$-sized classical algorithm assigning appropriate edge weights to convert a given $\mathrm{FBDD}_f$ to the $\mathrm{WFBDD}_f$ which describes the uniform superposition.
Our construction exponentially improves the number of \textsc{CNOT} gates used by another BDD-based QSP~\cite{Mozafari2020-lq}, which also aims to prepare a uniform superposition over a Boolean function.

Second, to prepare another superposition state encoding a Boolean function, when the Boolean function is represented by a $\mathrm{FBDD}_f$ with $N$ nodes, we modified our circuit-construction procedure to construct a controlled phase operation $\ket{x} \mapsto e^{i\theta f(x)} \ket{x}$ using at most $4(N-1)$ Toffoli gates and $N-1$ ancillary qubits. This result provides an instantiation of a phase oracle for a Boolean function, from which we can prepare the quantum state $\sum_{x \in \{0,1\}^n} e^{i\theta f(x)} \ket{x}$ trivially. 

Third, we give one example of $\mathrm{WFBDD}_f$ with $O(n^2)$ nodes describing an $n$-qubit quantum state for which the cost expression for the number of amplitude-amplification steps in the black-box QSP of Ref.~\cite{Bausch2022fastblackboxquantum} is exponential in $n$. Thus, this example shows that the QSP procedures are complementary to each other.

Furthermore, we showed that our state-preparation procedure can be applied to the block encodings introduced in Ref.~\cite{10.1145/3313276.3316366}. 
Our research is significant in that it shows a different class of preparable quantum states characterized by weighted graphs encoding Boolean functions than those of previous works~\cite{Zhang2022-bb,10.1109/TCAD.2023.3244885,Rosenthal2021-pj}.
As future work, it would be interesting to explore problems involving Boolean functions encoded by FBDDs and to investigate their solutions through quantum state preparation.

\begin{acknowledgments}
H.Y.\ was supported by JST PRESTO Grant Number JPMJPR201A, JPMJPR23FC, JSPS KAKENHI Grant Number JP23K19970, and MEXT Quantum Leap Flagship Program (MEXT QLEAP) JPMXS0118069605, JPMXS0120351339\@.
M.M. was supported by MEXT Q-LEAP Grant No. JPMXS0118069605, JSPS KAKENHI Grant Nos. 21H03394 and 23K21643, JST CREST Grant No. JPMJCR25I5, JST ASPIRE Grant No. JPMJAP25A3, and JST NEXUS Grant No. JPMJNX26C9@.

\end{acknowledgments}


\appendix

\section{Detailed proofs for WFBDD state preparation and uniform construction}
\label{Apdx:proofs}

\subsection{Proof details for Lemma~\ref{lem:1}}
\label{apdx:proof_lemma1}

As preparation, letting $N_c(u)$ denote the child nodes of the directed edges out from an internal node $u$, let us define the layers as follows:
\begin{itemize}
    \item[1.] $L_1=\{t_0,t_1\}$;
    \item[2.] for $i>1$,
    \begin{eqnarray}
        L_i=\left\{u\in V\setminus\{t_0,t_1\}\ \middle|\
        N_c(u)\subseteq L_1\cup\cdots\cup L_{i-1},\quad
        N_c(u)\cap L_{i-1}\neq\emptyset
        \right\}.
    \end{eqnarray}
\end{itemize}
Figure~\ref{fig:lemma1_layers} illustrates this layer construction.
Since the underlying graph is acyclic, every node belongs to one of these layers. Indeed, otherwise one could repeatedly move from an unlayered internal node to an unlayered child node, contradicting acyclicity and the fact that the terminal nodes are in $L_1$.

{\nolinenumbers
\begin{figure}[!b]
    \centering
    \includegraphics[keepaspectratio,scale=0.2]{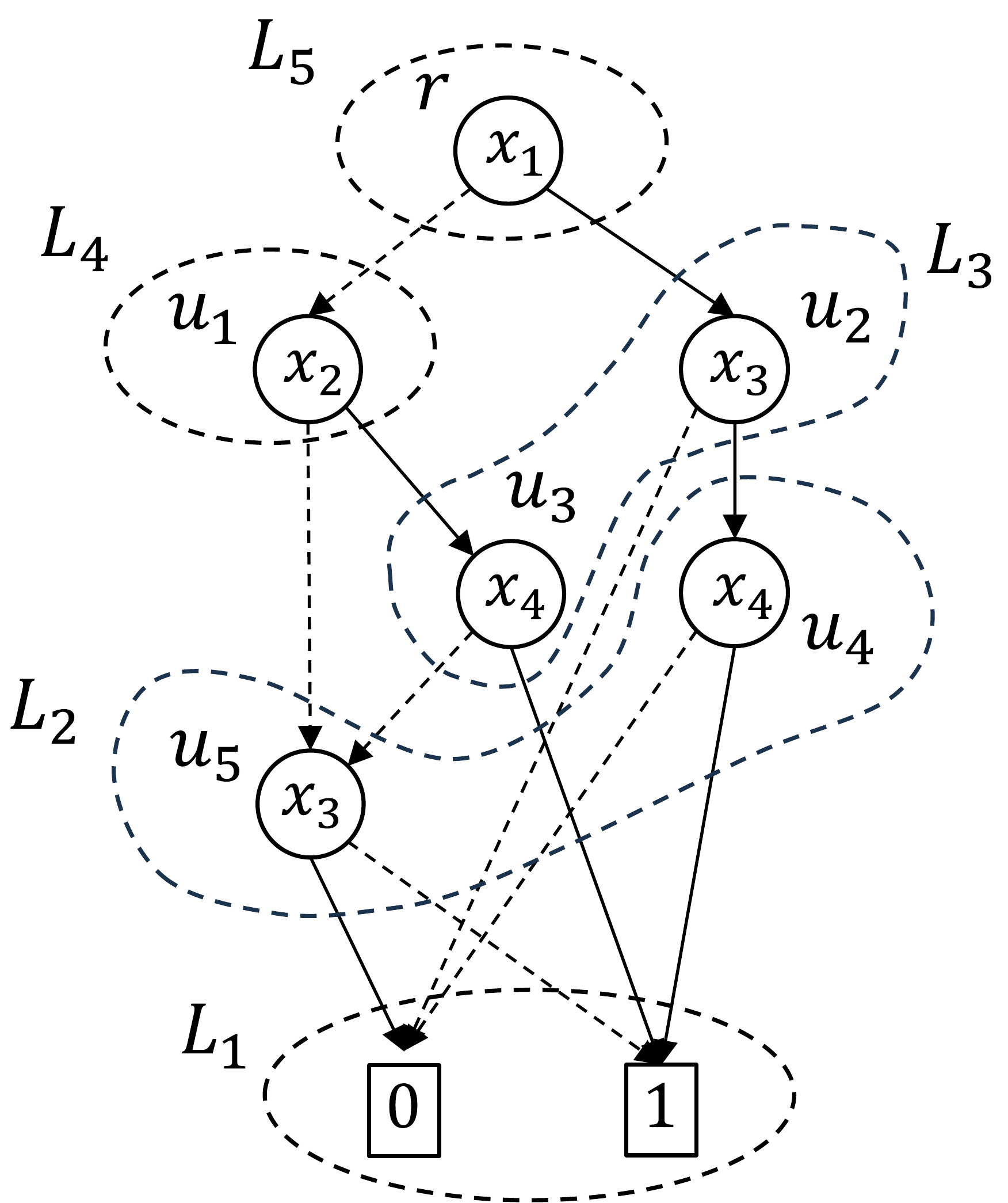}
    \caption{An example of layers in the proof of Lemma~\ref{lem:1}. The first layer $L_1$ is $\{t_0,t_1\}$ by definition. The second layer $L_2$ is $\{u_4,u_5\}$ since all the child nodes of the directed edges out from $u_4$ and $u_5$ belong to $L_1$. The remaining layers $L_3$, $L_4$, and $L_5$ can be defined sequentially in this way.}
    \label{fig:lemma1_layers}
\end{figure}
}

For a fixed root-to-node path $P_u=(V_u,E_u)$ ending at a node $u$, let $b_{P_u}(v) \in \{0,1\}$ be the branch selected at an internal node $v$ on this path. Let $B_{P_u}$ be the set of assignments $z\in\{0,1\}^n$ whose selected root-to-$t_1$ path contains the path $P_u$. If no such accepting path exists, then $B_{P_u}=\emptyset$.
We prove, by induction over the layers, that for every such path $P_u$,
\begin{eqnarray}
    \sum_{z\in B_{P_u}} |\alpha(z)|^2
    =
    \prod_{v\in V_u\setminus\{u,t_0,t_1\}}
    |w_{b_{P_u}(v)}(v)|^2.
    \label{eq:lemma1_layer_claim_revised}
\end{eqnarray}
The empty product is taken to be one. Notice that the right-hand side is the probability mass carried by the prefix $P_u$.

We first consider the terminal layer. If $u=t_1$, then all assignments in $B_{P_u}$ share the same accepting root-to-$t_1$ path $P_u$, while the variables not queried on this path are free. The number of such free variables is $n-|V_u|+1$. Therefore Def.~\ref{def:q_state_wfbdd} gives
\begin{eqnarray}
    \sum_{z\in B_{P_u}} |\alpha(z)|^2
    &=&
    2^{n-|V_u|+1}
    \frac{1}{2^{n-|V_u|+1}}
    \prod_{v\in V_u\setminus\{t_0,t_1\}}
    |w_{b_{P_u}(v)}(v)|^2 \\
    &=&
    \prod_{v\in V_u\setminus\{t_1,t_0\}}
    |w_{b_{P_u}(v)}(v)|^2,
\end{eqnarray}
which is Eq.~(\ref{eq:lemma1_layer_claim_revised}). If $u=t_0$, then $B_{P_u}=\emptyset$. Moreover, along any root-to-$t_0$ path that has nonzero prefix mass before reaching $t_0$, there is a first edge that is not contained in any root-to-$t_1$ path. By Condition~3 of Def.~\ref{def:qsp_admissible_wfbdd}, the corresponding branch weight is zero. Hence both sides of Eq.~(\ref{eq:lemma1_layer_claim_revised}) are zero.

Assume that Eq.~(\ref{eq:lemma1_layer_claim_revised}) holds for all nodes in $L_1\cup\cdots\cup L_i$, and take an internal node $u\in L_{i+1}$. For a fixed root-to-$u$ path $P_u$, let $P_{h_b(u)}$ be the path obtained by appending the $b$-edge of $u$ to $P_u$. The accepting assignments extending $P_u$ are the disjoint union of those extending $P_{h_0(u)}$ and those extending $P_{h_1(u)}$. By the induction hypothesis,
\begin{eqnarray}
    \sum_{z\in B_{P_u}} |\alpha(z)|^2
    &=&
    \sum_{b\in\{0,1\}}
    \sum_{z\in B_{P_{h_b(u)}}} |\alpha(z)|^2 \\
    &=&
    \prod_{v\in V_u\setminus\{u,t_0,t_1\}}
    |w_{b_{P_u}(v)}(v)|^2
    \left(|w_0(u)|^2+|w_1(u)|^2\right).
\end{eqnarray}
If the prefix $P_u$ carries nonzero mass, then $u$ lies on a root-to-$t_1$ path and is therefore relevant; Condition~2 of Def.~\ref{def:qsp_admissible_wfbdd} gives $|w_0(u)|^2+|w_1(u)|^2=1$. If the prefix already carries zero mass, the equality is trivial. Thus Eq.~(\ref{eq:lemma1_layer_claim_revised}) holds for $u$, completing the induction.

Finally, applying Eq.~(\ref{eq:lemma1_layer_claim_revised}) to the root path $P_r$ gives
\begin{eqnarray}
    \sum_{z:P(z)\text{ reaches }t_1} |\alpha(z)|^2=1.
\end{eqnarray}
Assignments whose selected path reaches $t_0$ have amplitude zero by Def.~\ref{def:q_state_wfbdd}. Hence $\sum_z |\alpha(z)|^2=1$. Also, for each $z$, $|\alpha(z)|^2\in[0,1]$ because the relevant outgoing pairs are locally normalized and the free-variable prefactor is at most one. This proves Lemma~\ref{lem:1}.

\subsection{Proof details for Theorem~\ref{thm:WS}}
\label{apdx:proof_theorem1}

First of all, we show that Algorithm~\ref{alg:WFBDD_SG} gives a stack $S$ of all internal nodes of $\textrm{WFBDD}_f$ with a linear total ordering such that for every directed edge $(u,v)$ out from an internal node $u$ to another internal node $v$, $u$ comes before $v$ in the ordering, {\it i.e.}, a topological sort.
In Algorithm~\ref{alg:WFBDD_SG}, a set $Q=\{ r \}$ consisting of the root $r$ is prepared (Line~\ref{alg:SG2}), $r$ is removed from $Q$ (Line~\ref{alg:SG8}) and pushed to $S$ (Line~\ref{alg:SG9}).
Since $r$ has indegree $0$, $S$ follows the ordering.
Next, for each $b$-edge $e_b(r)$ out from $r$ to its child node $h_b(r)$, we remove $e_b(r)$ from $E$ and check the indegree of $h_b(r)$ and whether $h_b(r)$ is a terminal node (Lines~\ref{alg:SG12}-\ref{alg:SG14}).
If the indegree is zero and $h_b(u) \neq t_0, t_1$, then $h_b(r)$ is added to $Q$ (Line~\ref{alg:SG15}).
The added node $h_b(r)$ is removed from $Q$ and pushed to $S$ in some round of the while loop. The child node $h_b(r)$ follows the ordering because the indegree zero means that all the directed edges whose child node is $h_b(r)$ have been removed from $E$, and the tail nodes of those directed edges have been pushed to $S$.
Repeating this process, we order all the nodes. If there is an unordered internal node $u$, then the indegree of $u$ must be nonzero, which means another unordered internal node.
Repeating this discussion leads to the conclusion that the root $r$ would be unordered, which contradicts the fact that $r$ is first pushed to $S$. Let $\prec$ denote the total ordering.

We show that Algorithm~\ref{alg:WFBDD_SG} gives the following state by induction:
\begin{eqnarray}
    \ket{\psi_u}
    =
    \sum_{z\in\{0,1\}^n} \alpha_u(z)\ket{z}\ket{g_u(z)},
    \label{eq:state_invariant_original_style_revised}
\end{eqnarray}
where $\alpha_u(z)$ and $g_u(z)$ are defined as follows. Define $V_{\preceq u}:=\{v\in V\setminus\{t_0,t_1\}\mid v\preceq u\}$ as a set of an internal node $u$ and all the preceding nodes to $u$, and $P(z)=(V(z),E(z))$ as the path derived by starting from the root $r$ and selecting the $z_i$-edge in each internal node until reaching a terminal node. Then, $\alpha_u(z)$ is defined as
\begin{eqnarray}
    \alpha_u(z)
    =
    \frac{1}{\sqrt{2^{n-|V_{\preceq u}\cap V(z)|}}}
    \prod_{v\in V_{\preceq u}\cap V(z)} w_{z_{i(v)}}(v),
    \label{eq:alpha_u_original_style_revised}
\end{eqnarray}
where terminal nodes are not included in the product. If the selected prefix has already taken a zero-weight branch, this formula gives $\alpha_u(z)=0$.
We define $g_u(z)$ as the $(|V|-2)$-bit sequence such that each bit is assigned to each node $v\in V$ except the terminal nodes $t_0$ and $t_1$, and each bit value assigned to the node $v$ is $1$ if $v$ has been activated along the selected path up to that point, and it is $0$ otherwise.

When $u=r$, the controlled-$H$ gate and the controlled-$U(r)$ gate controlled by the qubit $r$ set to $\ket{1}$ are applied to the qubit $x_{i(r)}$ set to $\ket{+}$ (Lines~\ref{alg:SG3}, \ref{alg:SG4}, \ref{alg:SG10}, and \ref{alg:SG11}).
Since $r$ is relevant and the WFBDD is QSP-admissible, the outgoing pair $(w_0(r),w_1(r))$ is locally normalized, and the unitary in Line~\ref{alg:SG11} of Algorithm~\ref{alg:WFBDD_SG} maps $\ket{0}$ to
\begin{eqnarray}
    w_0(r)\ket{0}+w_1(r)\ket{1}.
\end{eqnarray}
Thus, Eq.~(\ref{eq:alpha_u_original_style_revised}) holds for the root $r$,
\begin{eqnarray}
    \alpha_r(z)
    =
    \frac{1}{\sqrt{2^{n-1}}}
    w_{z_{i(r)}}(r).
\end{eqnarray}
Furthermore, for the child node $h_0(r)\neq t_0,t_1$, the qubit $x_{i(r)}$ is flipped, and the Toffoli gate controlled by the qubits $r$ and $x_{i(r)}$ is applied to the qubit $h_0(r)$. Then, the qubit $x_{i(r)}$ is flipped again (Lines~\ref{alg:SG17}-\ref{alg:SG20}). For the child node $h_1(r)\neq t_0,t_1$, the same operation is also performed except the two flipping operations. Remembering that the qubit $r$ is set to $\ket{1}$ (Line~\ref{alg:SG4}), we derive that the invariant $g_r(z)$ holds.

Assume that $\alpha_u(z)$ and $g_u(z)$ satisfy their definitions for an internal node $u$ and any internal node preceding the node, and consider the internal node $v$ following the node $u$ in the ordering. For an arbitrary path $P(z)=(V(z),E(z))$ from the root $r$, through $v$, and to one of the terminal nodes, let $e_b(u)=(u,v)\in E(z)$ and $e_{b'}(v)=(v,w)\in E(z)$ denote directed edges into and out from the node $v$, respectively.
If $v\notin V(z)$, then the register assigned to $v$ is $0$ on the component indexed by $z$, so the controlled gates at $v$ do not act on that component. If $v\in V(z)$ and the component has nonzero amplitude, then the register assigned to $v$ is $1$. Since the diagram is an FBDD, the variable $x_{i(v)}$ has not appeared earlier on the same path, and its qubit is still in the uniform state $\ket{+}$ before the controlled-$H$ gate at $v$ is applied.
Applying both of the controlled-$H$ gate and the controlled-$U(v)$ gate, controlled by the qubit $v$, to qubit $x_{i(v)}$ set to $\ket{+}$ (Lines~\ref{alg:SG3}, \ref{alg:SG10}, and \ref{alg:SG11}), we obtain that
\begin{eqnarray}
    \alpha_v(z)
    &=&
    \alpha_u(z)\sqrt{2}\,w_{z_{i(v)}}(v) \\
    &=&
    \frac{1}{\sqrt{2^{n-|V_{\preceq v}\cap V(z)|}}}
    \prod_{v'\in V_{\preceq v}\cap V(z)}w_{z_{i(v')}}(v').
\end{eqnarray}
If the branch selected at $v$ lies outside all root-to-$t_1$ paths, then its weight is zero by Condition~3 of Def.~\ref{def:qsp_admissible_wfbdd}, and the same formula gives zero amplitude. If an irrelevant node is ever reached, the amplitude of that component is already zero, so the arbitrary choice of $U(v)$ at irrelevant nodes does not affect the final state.
Furthermore, when $b'=0$, the qubit $x_{i(v)}$ is flipped, and the Toffoli gate controlled by the qubits $v$ and $x_{i(v)}$ is applied to the qubit $w$, which is the child node $h_{b'}(v)$ of the $b'$-edge $e_{b'}(v)$. Then, the qubit $x_{i(v)}$ is flipped again (Lines~\ref{alg:SG17}-\ref{alg:SG20}). When $b'=1$, the same operation is also performed except the two flipping operations. Thus, $g_v(z)$ also satisfies its definition. This completes the induction.

After all internal nodes have been processed, Eq.~(\ref{eq:alpha_u_original_style_revised}) gives the product of all selected branch weights on the root-to-terminal path. If $P(z)$ reaches $t_1$, this is exactly the accepting-path amplitude in Def.~\ref{def:q_state_wfbdd}, because the number of queried variables is $|V(z)|-1$ when the terminal node is included in $V(z)$. If $P(z)$ reaches $t_0$, then the path contains a first edge that lies on no root-to-$t_1$ path, and Condition~3 of Def.~\ref{def:qsp_admissible_wfbdd} makes the corresponding branch weight zero. Thus the amplitude agrees with the rejecting case of Def.~\ref{def:q_state_wfbdd}.

The undo stage pops the nodes from $S$ in the reverse topological order and reapplies the same child-recording Toffoli gates. These Toffoli gates are self-inverse. Because parents are undone after their children, the controls needed to uncompute each selected child register are still present when the corresponding Toffoli is applied. Hence all internal-node registers except the root register are returned to $\ket{0}$, while the root register remains in its initial state $\ket{1}$. The variable-register amplitudes are unchanged during this undo stage, and the final joint state is
\begin{eqnarray}
    \left(\sum_{z\in\{0,1\}^n}\alpha(z)\ket{z}\right)
    \ket{1}_r
    \bigotimes_{v\in V\setminus\{r,t_0,t_1\}}\ket{0}_v,
\end{eqnarray}
where $\alpha(z)$ is the amplitude in Def.~\ref{def:q_state_wfbdd}. Therefore the variable register contains the desired quantum state.

Finally, each internal node contributes one controlled-$H$ gate and one controlled-$U$ gate, and each outgoing edge contributes only a constant number of Toffoli gates in the compute and uncompute stages. Since a BDD has two outgoing edges per internal node, the number of such gates is $O(|V|)$. Using a constant-size decomposition of Toffoli gates into single- and two-qubit gates gives the stated $O(|V|)$ gate count, and the internal-node registers require $|V|-2$ ancillary qubits, including the root register.

\subsection{Proof details for Theorem~\ref{thm:uniform_wfbdd}}
\label{apdx:proof_uniform}

We prove Theorem~\ref{thm:uniform_wfbdd} by following the steps of the uniform-weight construction: first the processing order, then the consistency of the variable sets, then the weight assignment, and finally the uniformity of the resulting amplitudes.

First, we justify the processing order used in Line~\ref{alg:UWD:line4}.
A reverse topological order can be obtained as follows. Starting from the root $r$, maintain a queue $Q$ of internal nodes whose incoming edges from unprocessed internal nodes have all been removed. Initialize $Q=\{r\}$ and an empty list $T$. Whenever a node $u$ is removed from $Q$, append $u$ to $T$, remove the two outgoing edges of $u$, and add an internal child $h_b(u)$ to $Q$ once all its incoming edges have been removed. Since the FBDD is a finite directed acyclic graph, this procedure lists all internal nodes in a topological order, in which every parent appears before its internal children. Reversing the list $T$ therefore gives an order in which every internal child is processed before its parent. Terminal nodes are already initialized by $R(t_0)=R(t_1)=X$ and $M(t_0)=0,M(t_1)=1$. Hence, when Algorithm~\ref{alg:UWD} processes an internal node $u$, the quantities associated with both children $h_0(u)$ and $h_1(u)$ have already been computed.

Second, we explain the consistency of the remaining-variable sets. The terminal convention $R(t_0)=R(t_1)=X$ means that, before any query is specified below a terminal, all variables are formally available as free variables. Suppose that the sets for the two children of an internal node $u$ have already been assigned consistently. The common set
\begin{eqnarray}
    Y(u)=R(h_0(u))\cap R(h_1(u))
\end{eqnarray}
contains the variables that can be kept common to both branches out of $u$. Since the diagram is an FBDD, the variable $x_{i(u)}$ queried at $u$ cannot appear again below either child and hence it belongs to this common set. Thus
\begin{eqnarray}
    R(u)=Y(u)\setminus\{x_{i(u)}\}
\end{eqnarray}
is the remaining-variable set after the query at $u$ has been accounted for. For a branch $u\to h_b(u)$, the variables in $R(h_b(u))\setminus Y(u)$ are precisely the branch-dependent variables that are skipped on this branch relative to the common variables kept at $u$. This is the self-consistency condition for the remaining-variable sets, expressed in terms of $R(u)$ and $Y(u)$.

Third, we consider the effective branch counts. By the preceding consistency, each accepting continuation below the child $h_b(u)$ contributes $2^{|R(h_b(u))\setminus Y(u)|}$ assignments through the branch $u\to h_b(u)$, because the variables in $R(h_b(u))\setminus Y(u)$ are skipped on that branch. Therefore
\begin{eqnarray}
    M_b(u)=2^{|R(h_b(u))\setminus Y(u)|}M(h_b(u))
\end{eqnarray}
counts the accepting continuations through the $b$-edge after the common remaining variables have been factored out. The two branches correspond to the disjoint choices $x_{i(u)}=0$ and $x_{i(u)}=1$, so
\begin{eqnarray}
    M(u)=M_0(u)+M_1(u).
\end{eqnarray}
More generally, after the variables in $R(u)$ have been factored out, $M(u)$ counts the accepting patterns in the subdiagram rooted at $u$. This statement is immediate for $t_0$ and $t_1$ from $M(t_0)=0$ and $M(t_1)=1$, and the branch-count formula above gives the induction step because the two branches correspond to disjoint values of $x_{i(u)}$. Applying this recursive counting statement to the root gives
\begin{eqnarray}
    |f|=2^{|R(r)|}M(r).
    \label{eq:root_count_original_style_revised}
\end{eqnarray}

For every relevant internal node $u$, there is at least one accepting continuation below $u$ and hence $M(u)>0$. The assigned weights satisfy
\begin{eqnarray}
    |w_0(u)|^2+|w_1(u)|^2
    =
    \frac{M_0(u)}{M(u)}+\frac{M_1(u)}{M(u)}=1.
\end{eqnarray}
Thus the outgoing weight pair at every relevant internal node is locally normalized. If a branch $u\to h_b(u)$ has no accepting continuation, then $M_b(u)=0$ and the assigned weight $w_b(u)$ is zero. Therefore the constructed WFBDD satisfies the accepting-support condition in Def.~\ref{def:qsp_admissible_wfbdd}. Since $|f|>0$, the underlying FBDD has at least one accepting input, and the constructed WFBDD is QSP-admissible.

It remains to show that the associated quantum-state semantics is the uniform superposition. Fix an accepting path
\begin{eqnarray}
    P:r=u_1\to u_2\to\cdots\to u_{\ell+1}=t_1,
\end{eqnarray}
and let $b_j$ be the branch selected at $u_j$, so that $u_{j+1}=h_{b_j}(u_j)$. Put
\begin{eqnarray}
    s_j=|R(u_{j+1})\setminus Y(u_j)|.
\end{eqnarray}
The self-consistency of the remaining-variable sets implies that the variables not queried on the accepting path $P$ are exactly the variables in $R(r)$ together with the branch-dependent skipped variables counted by the $s_j$'s. Hence
\begin{eqnarray}
    n-|V_P|+1=|R(r)|+\sum_{j=1}^{\ell}s_j,
    \label{eq:free_variables_original_style_revised}
\end{eqnarray}
where $V_P$ is the set of nodes on $P$, including the terminal node.

Using the definition of the weights and the effective branch counts, the product of the weights along the accepting path telescopes:
\begin{eqnarray}
    \prod_{j=1}^{\ell} w_{b_j}(u_j)
    &=&
    \prod_{j=1}^{\ell}\sqrt{\frac{M_{b_j}(u_j)}{M(u_j)}} \\
    &=&
    \prod_{j=1}^{\ell}\sqrt{\frac{2^{s_j}M(u_{j+1})}{M(u_j)}}
    =
    \sqrt{\frac{2^{\sum_{j=1}^{\ell}s_j}}{M(r)}},
    \label{eq:weight_telescoping_original_style_revised}
\end{eqnarray}
where we used $M(t_1)=1$ in the last equality. Therefore, for every assignment $z$ whose selected path is this accepting path, Def.~\ref{def:q_state_wfbdd} gives
\begin{eqnarray}
    \alpha(z)
    &=&
    \frac{1}{\sqrt{2^{n-|V_P|+1}}}
    \prod_{j=1}^{\ell} w_{b_j}(u_j) \\
    &=&
    \frac{1}{\sqrt{2^{|R(r)|+\sum_j s_j}}}
    \sqrt{\frac{2^{\sum_j s_j}}{M(r)}}
    =
    \frac{1}{\sqrt{2^{|R(r)|}M(r)}}
    =
    \frac{1}{\sqrt{|f|}},
\end{eqnarray}
where the last equality follows from Eq.~(\ref{eq:root_count_original_style_revised}). If an assignment is rejected, its selected path reaches $t_0$, and Def.~\ref{def:q_state_wfbdd} assigns amplitude zero. Thus the quantum-state semantics of the constructed WFBDD is exactly $\ket{\phi_f}$.

Finally, Algorithm~\ref{alg:UWD} processes each internal node once and inspects its two outgoing edges. Hence it uses $O(|V|)$ node/edge queries to the supplied FBDD. If the remaining-variable sets are stored explicitly, the intersections and set differences over $n$ variables require $O(n|V|)$ classical set-processing time.

\section{Quantum Random Access Memory (QRAM)} \label{Apdx:A}

This appendix discusses one possible implementation of coherent access to classical BDD data. The qROM/QRAM resources described here should be added separately when such coherent access is required; they are not included in the state-preparation circuit-size statements in the main text.

To access the classical information in the BDD data structure quantumly, we consider a quantum random access memory (QRAM) that stores quantum-accessible classical bits. 
In the implicit model of QRAM in \cite{Matteo2019FaultTolerantRE}, locations in memory are addressed by $n$-bit strings $x_1 x_2 \cdots x_n$ and are queried by inputting the associated computational-basis state into a circuit:
\begin{eqnarray}
\ket{x_1 x_2 \cdots x_n} \ket{0} \mapsto \ket{x_1 x_2 \cdots x_n} \ket{b_{x_1 x_2 \cdots x_n}},
\end{eqnarray}
where $b_{x_1 x_2 \cdots x_n} \in \{ 0,1 \}$ is the stored value. 
Formally, one can realize the QRAM by implementing a sequence of mixed-polarity multi-controlled Toffolis (MPMCTs) conditioned on the control bits representing the memory address of a $1$.
There are two known decompositions of quantum circuits into 1- and 2-qubit operations of the Clifford+$T$ gate set. One is a large-depth, small-width circuit. The other is a small-depth, large-width circuit~\cite{Matteo2019FaultTolerantRE}.

Let $2^q$ be the number of $1$s in a database. 
To derive the small-depth, large-width circuit, first, the address bits input to the first registers of qubits are copied to the other ancilla registers of qubits by using a log-depth cascade of \textsc{CNOT}s. 
Second, each register performs an MPMCT whose target is one of the additional $2^q$ qubits.
The state of these additional qubits is prepared in a superposition over even-parity states.
Finally, one computes the parity, copies it to an additional ancilla qubit, copies it back, and then uncomputes the parity and the address fanout (See Fig. 4 in \cite{Matteo2019FaultTolerantRE}).
From \cite{Matteo2019FaultTolerantRE}, the required number of qubits (the width of the circuit) is $O(n2^q)$, and the required depth of the circuit  scales linearly in both $n$ and $q$.

\bibliography{qnl}

\end{document}